\documentclass[11pt]{amsart}

\pdfoutput = 1

\usepackage[english]{babel}
\usepackage[autostyle]{csquotes}
\usepackage{amsmath,amssymb,amsthm}
\usepackage{amsfonts}
\usepackage{amsxtra}

\usepackage{array,dcolumn}
\usepackage{graphicx}
\usepackage[hyperfootnotes=true,
        pdffitwindow=true,
        plainpages=false,
        pdfpagelabels=true,
        pdfpagemode=UseOutlines,
        pdfpagelayout=SinglePage,
                hyperindex,]{hyperref}

\usepackage{lmodern}
\usepackage[T1]{fontenc}
\usepackage[utf8]{inputenc}
\usepackage[activate={true,nocompatibility},spacing,kerning]{microtype}
\usepackage{bm}
\usepackage{bbm}
\usepackage[sc,osf]{mathpazo}
\microtypecontext{spacing=nonfrench}

  \calclayout

\newcommand{\mathsym}[1]{{}}
\newcommand{\unicode}[1]{{}}

\theoremstyle{plain}
\newtheorem{theorem}{Theorem}

\newtheorem{corollary}[theorem]{Corollary}
\newtheorem{proposition}[theorem]{Proposition}

\theoremstyle{definition}

\theoremstyle{remark}
\newtheorem{remark}[theorem]{Remark}

\renewcommand{\geq}{\geqslant}

\numberwithin{equation}{section}

\numberwithin{theorem}{section}

\begin{document}


\title[]{Expanding the Fourier transform of the scaled 
circular Jacobi $\beta$ ensemble density}

\author{Peter J. Forrester}
\address{School of Mathematics and Statistics, 
University of Melbourne, Victoria 3010, Australia}
\email{pjforr@unimelb.edu.au}

\author{Bo-Jian Shen}
\address{School of Mathematical Sciences, Shanghai Jiaotong University, CMA-Shanghai, Shanghai 200240, People's Republic of China}
\email{JOHN-EINSTEIN@sjtu.edu.cn}

\date{\today}


\begin{abstract}
The family of circular Jacobi $\beta$ ensembles has a singularity of a type associated with Fisher and Hartwig in the theory of Toeplitz determinants. Our interest is in the Fourier transform of the corresponding bulk scaled spectral density about this singularity, expanded as a series in the Fourier variable. Various integrability aspects of the  circular Jacobi $\beta$ ensemble are used for this purpose. These include linear differential equations satisfied by the scaled spectral density for $\beta = 2$ and $\beta = 4$, and the loop equation hierarchy. The polynomials in the variable $u=2/\beta$ which occur in the expansion coefficents are found to have special properties analogous to those known for the structure function of the circular $\beta$ ensemble, specifically in relation to the zeros lying on the unit circle $|u|=1$ and interlacing. Comparison is also made with known results for the expanded Fourier transform of the density about a guest charge in the two-dimensional one-component plasma.
\end{abstract}


\maketitle


\section{Introduction}\label{S1}
In random matrix theory, the joint eigenvalue probability density function (PDF) proportional to
\begin{equation}\label{1.1}
\prod_{l=1}^N w^{(\rm cJ)}(\theta_l) \prod_{1 \le j < k \le N} | e^{i \theta_k} - e^{i \theta_j} |^\beta, \qquad w^{(\rm cJ)}(\theta) = e^{q \theta} |1 + e^{i \theta} |^{\beta p} \: \: (-\pi < \theta \le \pi),
\end{equation}
is referred to as the generalised circular Jacobi $\beta$ ensemble. This reduces to the PDF specifying the circular $\beta$ ensemble (see \cite[\S 2.8]{Fo10}) upon setting $p=q=0$.
With $q=0$ but $p \ne 0$, an interpretation of the factor $ |1 + e^{i \theta} |^{ \beta p}$ is as being due to a spectrum singularity of degeneracy $p$ at $\theta = \pm \pi$, and the resulting PDF has
been termed the circular Jacobi ensemble; see \cite[\S 3.9]{Fo10}. With $q \ne 0$ also, the factor $e^{q \theta}$ is discontinuous comparing values at $\theta = \pm \pi$, extending the spectrum
singularity to the class introduced by Fisher and Hartwig in the theory of asymptotics of Toeplitz determinants \cite{FH68}. We remark that generally the study of Fisher-Hartwig singularities in random matrix theory and its applications is active to this present day, with recent references including \cite{DIK13,FS16,BWW18,WW19,XZ20,CG21,DS22,BF22,CGMY22,BC22,AAW23}. The naming ``circular Jacobi" comes from the fact
that the orthogonal polynomials in $z = e^{i \theta}$ with respect to the inner product $(f,g) = \int_{-\pi}^\pi f(\theta) \overline{g(\theta)} w^{(\rm cJ)}(\theta) |_{\beta = 2}
\, d \theta$ are given by the 
hypergeometric polynomials $\{ {}_2 F_1 (-n,b,b+\bar{b};1-z) \}$, where $b = p-iq$, which in term
can be written in terms of 
Jacobi polynomials \cite{SR10}.

The PDF (\ref{1.1}) is an example of a $\beta$ ensemble with the eigenvalues supported on the unit circle in the complex plane. In contrast, a $\beta$ ensemble on the real line
is specified by a joint eigenvalue PDF of the form
\begin{equation}\label{1.2}
\prod_{l=1}^N w(x_l) \prod_{1 \le j < k \le N} | x_k - x_j |^\beta.
\end{equation}
Generally an Hermitian matrix $H$ can be constructed from a unitary matrix $U$ by way of the Cayley transform
$$
H = i {\mathbb I_N - U \over \mathbb I_N + U }.
$$
This was introduced into random matrix theory by Hua \cite{Hu63}. The corresponding mapping of the eigenvalues is
\begin{equation}\label{1.2a}
x = i {1 - e^{i \theta} \over 1 + e^{i \theta} } = \tan (\theta/2),
\end{equation}
which is recognised as specifying a stereographic projection between the unit circle and the real line.
Applying the inverse of this mapping, up to proportionality the PDF (\ref{1.1}) transforms to the functional form (\ref{1.2}) with
\begin{equation}\label{1.2b}
w(x) = {1 \over (1 - i x)^c (1 + i x)^{\overline{c}}}, \qquad c = \beta(N + p -1)/2 + 1  - iq,
\end{equation}
which is then referred to as the generalised Cauchy $\beta$ ensemble; see \cite[Eq.~(3.124)]{Fo10}.

It has recently been pointed out that there is also a relationship between the generalised Cauchy $\beta$ ensemble and the
Jacobi $\beta$ ensemble, defined by the PDF (\ref{1.2}) with weight
\begin{equation}\label{1.2c}
w(x) = x^{\lambda_1} (1 - x)^{\lambda_2} \chi_{0 < x < 1},
\end{equation}
where $\chi_A = 1$ for $A$ true, $\chi_A = 0$ otherwise. Thus for $f$ a multivariable polynomial, one has that
\cite{FR21}
\begin{equation}\label{1.2d}
\langle f(1 - i x_1, \dots, 1 - i x_N) \rangle^{(\rm Cy)} = \langle f(2x_1,\dots, 2x_N) \rangle^{(\rm J)} \Big |_{\lambda_1 = - \beta(N+p-1)/2 - 1 + iq\atop
\lambda_2 = - \beta(N+p-1)/2 - 1 - i  q },
\end{equation}
where the average on the RHS is to be understood in the sense of analytic continuation. 
The identity (\ref{1.2d}) was used in \cite{FR21} to obtain linear differential
equations satisfied by the density in the Cauchy ensemble with $\beta = 1,2$ and 4, and to study the corresponding moments.
This line of study was initiated in  \cite{ABGS20}.
Another application of (\ref{1.2d}) was given in \cite{Fo21}, where it was used to study the distribution of the trace in the Cauchy $\beta$ ensemble; see also
\cite{AGS21}.

Applying the mapping (\ref{1.2a}) and recalling the theory of the previous paragraph gives the relationship between averages in the circular
Jacobi and Jacobi ensembles
\begin{equation}\label{1.6a}
\Big \langle f \Big ( {1 \over 1 + e^{i \theta_1}},  \dots,  {1 \over 1 + e^{i \theta_N}} \Big ) \Big \rangle^{(\rm cJ)} = \langle f(x_1,\dots, x_N) \rangle^{(\rm J)} \Big |_{\lambda_1 = - \beta(N+p-1)/2 - 1 + iq\atop
\lambda_2 = - \beta(N+p-1)/2 - 1 - i  q }.
\end{equation}
Note that unlike (\ref{1.2d}), (\ref{1.6a}) no longer inter-relates averages of polynomials. On the other hand, there is a further
inter-relation which does have this property, namely \cite[corollary of Prop.~3.9.1]{Fo10}
\begin{equation}\label{1.6b}
\langle f( -e^{i \theta_1}, \dots, -e^{i \theta_N}) \rangle^{(\rm cJ)} = \langle f( x_1,\dots,  x_N) \rangle^{(\rm J)} \Big |_{\lambda_1 = - \beta(N+p-1)/2 - 1 - iq\atop
\lambda_2 = \beta p \hspace*{1.85cm}},
\end{equation}
again valid in the sense of analytic continuation. Note that shifting the integration variables $\theta_j \mapsto \theta_j + \pi$ on the LHS of (\ref{1.6b}) removes the minus signs in the arguments of $f$ to give
\begin{equation}\label{1.6b+}
\langle f( e^{i \theta_1}, \dots, e^{i \theta_N}) \rangle^{(\widetilde{\rm cJ})} = \langle f( x_1,\dots,  x_N) \rangle^{(\rm J)} \Big |_{\lambda_1 = - \beta(N+p-1)/2 - 1 + iq\atop
\lambda_2 = \beta p \hspace*{1.85cm}},
\end{equation}
where $\widetilde{\rm cJ}$ refers to the circular Jacobi $\beta$ ensemble now defined on $0\le \theta_j < 2 \pi$ with weight
\begin{equation}\label{1.6c}
w^{(\widetilde{\rm cJ})}(\theta) = e^{q (\theta - \pi)} |1 - e^{i \theta} |^{\beta p}.
\end{equation}

It is clear from the above discussion that known properties of the Cauchy and Jacobi $\beta$ ensembles have application in the study of
the generalised circular Jacobi $\beta$ ensemble. For example, some time ago \cite{FN01} the full set of $k$-point correlation functions for (\ref{1.1})
in the case $q=0$ and $\beta = 1,2,4$, and their scaling limit about the spectrum singularity, were deduced from knowledge of the corresponding $k$-point correlation functions in the case
of the Cauchy ensemble. More recently  \cite{FLT21} this analysis was repeated with the restriction to $q=0$ removed, and moreover the leading finite size corrections
to the scaling limit were determined.

 In the present paper we  focus our attention on integrable structures in relation to the density for the generalised circular Jacobi $\beta$ ensemble,
 $\rho_{(1),N}^{(\widetilde{\rm cJ)}}(\theta;\beta,p,q)$ say,
 and the analogue of the corresponding moments. The latter are the coefficients $\{ c_k^{(\widetilde{\rm cJ)}}(N,\beta,p,q) \}$ of the Fourier series for the density, 
\begin{equation}\label{1.2e} 
\rho_{(1),N}^{(\widetilde{\rm cJ})}(\theta;\beta,p,q) = {1 \over 2 \pi} \sum_{k = - \infty}^\infty c_k^{(\widetilde{\rm cJ)}}(N,\beta,p,q)   e^{-i k \theta}.
\end{equation}
Of particular interest is the scaled limit
\begin{equation}\label{1.2h} 
 c_\infty^{(\widetilde{\rm cJ})}(\tau;\beta,p,q)  =  \lim_{N, k \to \infty }      c_k^{(\widetilde{\rm cJ})}(N,\beta,p,q)  \Big |_{ \tau = 2 \pi k / N}, \qquad (\tau \ne 0),
   \end{equation} 
   which is expected to be well defined for general real $\tau$.
The integrable structures come from various sources, including the interplay with the Cauchy and Jacobi $\beta$ ensembles.

To see why this is of interest, we begin by noting that
 the density in the neighbourhood of the spectrum
 singularity admits a well defined large $N$ limit, first studied for $q=0$ and $\beta = 1,2$ and 4 in \cite{FN01}, and extended to all even $\beta$ 
 and general $q$ in \cite{Li17, FLT21}. 
 For general $\beta > 0$ this can also be regarded as a bulk scaling limit, with origin at the spectrum singularity,
 specified by taking $N \to \infty$ and the angular coordinates $\theta_j$ scaled $\theta_j \mapsto 2 \pi x_j/ L$ so that the mean spacing between
 eigenvalues $L/N$ is a constant, which we take to be unity by choosing $L=N$.
 The Fourier series
 for finite $N$, upon subtracting the constant $c_0^{(\rm cJ)} = 2 \pi N$, then tends to a well defined Fourier transform relating to $ c_\infty^{(\widetilde{\rm cJ})}$.
 Thus the scaling limit of the density in the  neighbourhood of the spectrum singularity is specified by
 \begin{equation}\label{1.2f} 
 \rho_{(1),\infty}^{(\rm cJ)}( x ;\beta,p,q)  :=  \lim_{N \to \infty} {2 \pi \over N}  \rho_{(1),N}^{(\rm cJ)}( -\pi \, {\rm sgn} (x)  + 2 \pi x/ N ;\beta,p,q), 
 \end{equation}
 or in the setting of the weight (\ref{1.6c})
 \begin{equation}\label{1.2f+} 
 \rho_{(1),\infty}^{(\widetilde{\rm cJ})}( x ;\beta,p,q)  :=  \lim_{N \to \infty} {2 \pi \over N}  \rho_{(1),N}^{(\widetilde{\rm cJ})}( 2\pi \chi_{x < 0}  + 2 \pi x/ N ;\beta,p,q).
 \end{equation}
 Applying this limit to the Fourier series (\ref{1.2e}) we obtain the Fourier transform expression
  \begin{equation}\label{1.2g} 
   \rho_{(1),\infty}^{(\widetilde{\rm cJ})}( x ;\beta,p,q)  - 1 =  {1 \over 2 \pi }   \int_{-\infty}^\infty c_\infty^{(\widetilde{\rm cJ})}(\tau;\beta,p,q) e^{-i \tau x} \, d \tau,
  \end{equation}  
  where $ c_\infty^{(\widetilde{\rm cJ})}$ is given by (\ref{1.2h}), and thus by Fourier inversion
    \begin{equation}\label{1.2g+} 
 c_\infty^{(\widetilde{\rm cJ})}(\tau;\beta,p,q)  =      \int_{-\infty}^\infty    \Big ( \rho_{(1),\infty}^{(\widetilde{\rm cJ})}( x ;\beta,p,q)  - 1 \Big )  e^{i \tau x} \, d x.
  \end{equation}

 In the case $q=0$, $p = 1$, (\ref{1.2g+}) can be identified with the bulk structure
 function for the limiting statistical state defined by (\ref{1.1}) with $p=q=0$ \cite{FJM00,Fo21a,CES21}. 
 Thus for finite $N$ the structure function $S_N(k;\beta)$ can be defined by the covariance of the pair of linear statistics
 $\sum_{j=1}^N  e^{i k \lambda_j}$, $\sum_{j=1}^N  e^{-i k \lambda_j}$,
 \begin{equation}\label{1.3}
 {S}_N(k;\beta) = \Big \langle \Big | \sum_{j=1}^N e^{i k \lambda_j} \Big |^2 \Big \rangle - \bigg |
   \Big \langle\sum_{j=1}^N e^{i k \lambda_j}  \Big \rangle \bigg |^2,
 \end{equation}
 where the averages are with respect to (\ref{1.1}) in the case $p=q=0$.
 In the bulk scaling limit,  this can be written in terms of the truncated (connected) two-particle correlation   function
 $\rho_{(2), \infty}^T(x_1, x_2)$ according to
 \begin{equation}\label{1.3a}
 {S}_\infty(\tau;\beta)   :=  \lim_{N, k \to \infty \atop \tau \: {\rm fixed}}  {S}_N(k;\beta)  \Big |_{\tau = 2 \pi k / N}= \int_{-\infty}^\infty \Big (  \rho_{(2), \infty}^T(x, 0) + \delta(x) \Big ) e^{i\tau x} \, dx.
  \end{equation}
  But by the definition of the truncated two-particle correlation   function in the case $p=q=0$ we have
  \begin{equation}\label{1.3b} 
  \rho_{(2),  \infty}^T(x, 0) =    \rho_{(1),\infty}^{(\widetilde{\rm cJ})}( x ;\beta,1,0)  - 1,
  \end{equation} 
  where the RHS  comes about since  one of the fixed eigenvalues involved in the definition of the two-particle correlation induces
  a spectrum singularity of degeneracy $1$. Comparing (\ref{1.3a}) with the substitution (\ref{1.3b}) to (\ref{1.2g+}) then tells us that
   \begin{equation}\label{1.3c}  
   {S}_\infty(\tau;\beta)  -1  =       c_\infty^{(\widetilde{\rm cJ})}(\tau;\beta,1,0). 
  \end{equation}   

In relation to the finite $N$ structure function for general random matrix ensembles, one remarks that the two-point quantity $ \langle  | \sum_{j=1}^N e^{i k \lambda_j}  |^2 \rangle$ on the RHS of (\ref{1.3}) has been the subject of some recent attention in literature applying random matrix theory to the study of many body quantum chaos \cite{CHLY17,TGS18} and the scrambling of information in black holes 
 \cite{C+17,CMS17}; many more references could be given. Of interest is its graphical shape as a function of $k$, which assuming a non-constant spectral density exhibits dip-ramp-plateau features as $k$ varies as a function of $N$, with corresponding dynamical significance. Following on from the applications, several theoretical works have quantified dip-ramp-plateau for various model random matrix ensembles. With such concerns being in common with present work when viewed broadly, we give a comprehensive list of references: \cite{Ok19,Fo21a,Fo21b,MH21,CES21,VG21,VG22,FKLZ23}

  Returning now to (\ref{1.3a}), we know from  \cite{FJM00}  that the quantity
\begin{equation}\label{tR2}
f(\tau;\beta) : = {\pi \beta \over |\tau|} S_\infty(\tau;\beta), \qquad 0 < \tau < {\rm min} \, (2 \pi, \pi \beta),
\end{equation}
extends to an analytic function of $\tau$ for $|\tau| < {\rm min} \, (2 \pi, \pi \beta)$ and furthermore
satisfies the functional equation
\begin{equation}\label{tR3}
f(\tau;\beta) = f \Big ( - {2\tau \over \beta} ; {4 \over \beta} \Big ).
\end{equation}
 The power series expansion of $f(\tau;\beta)$ in (\ref{tR3}) is known to have the form
   \cite{FJM00,WF15,Fo21a}
\begin{equation}\label{1.17a}
f(\tau;\beta) = 1 + \sum_{j=1}^\infty p_j(2/\beta) (\tau/ 2 \pi)^j
\end{equation} 
where
 \begin{equation}\label{1.17b}  
 p_j(u) = \left \{ \begin{array}{ll} (1 - u)^2 \sum_{l=0}^{j-2} b_{j,l} u^l \qquad (b_{j,0} = 1, \, b_{j,l} = b_{j,j-2-l}) \: j \, {\rm even} \\
 (1 - u) \sum_{l=0}^{j-1} b_{j,l} u^l \qquad (b_{j,0} = 1, \, b_{j,l} = b_{j,j-1-l}) \: j \, {\rm odd.}  \end{array} \right.
 \end{equation} 
Thus the coefficient of $\tau^j$ is a palindromic or anti-palindromic polynomial in $u=2/\beta$. Note that this structure is consistent
with (\ref{tR3}). The explicit form of $p_j(u)$ has been determined up to and including $j=9$ in \cite{FJM00,WF15}, while the explicit
form of $p_{10}(u)$ was calculated recently in \cite{Fo21a}.

The equality (\ref{1.3c}) tells us that $ (\pi \beta / |\tau| ) (c_\infty^{(\widetilde{\rm cJ})}(\tau;\beta,1,0) +1 )$ shares the same expansion
in $\tau$ as $f(\tau;\beta)$. In light of this, we ask about the small $\tau$ expansion of $c_\infty^{(\widetilde{\rm cJ})}(\tau;\beta,p,q)$ for general
$p$ and $q$. Analogous to the expansion (\ref{1.17a}), we will show that there are
coefficients $\{ h_j(\beta,p,q)\}$ and $\{ \tilde{h}_j(\beta,p,q) \}$ such that
 \begin{equation}\label{1.h}
 c_\infty^{(\widetilde{\rm cJ})}(\tau;\beta,p,q)   = \sum_{j=0}^\infty  h_j(\beta,p,q) \tau^{j} {\rm sgn} \, \tau
  +   \sum_{j=0}^\infty \tilde{h}_j(\beta,p,q)  \tau^{j} , \qquad |\tau| < |\tau^*|,
  \end{equation} 
  where $|\tau^*|$ is the radius of convergence of the series. We are most interested in properties of
  the coefficients as a function of $\beta, p$ and $q$.  As a beginning, we remark that the symmetry of 
  the discontinuous factor of (\ref{1.1}) in the neighbourhood of $\theta = \pm \pi$, namely $e^{q(-\pi+\theta)}$ $(\theta > 0)$
  and $e^{q(\pi+\theta)}$  $(\theta < 0)$,
 being  unchanged by the mapping $\theta \mapsto - \theta$  and $q \mapsto - q$,  implies
  $ \rho_{(1),\infty}^{(\rm cJ)}( x ;\beta,p,q) =  \rho_{(1),\infty}^{(\rm cJ)}( -x ;\beta,p,-q)$. Hence from (\ref{1.2g+}) we have the functional properties
   \begin{equation}\label{1.2gz}     
    h_{j}(\beta,p,q) =  (-1)^{j+1}  h_{j}(\beta,p,-q), \qquad  
     \tilde{h}_{j}(\beta,p,q) = (-1)^j   \tilde{h}_{j}(\beta,p,-q).
   \end{equation} 
   In addition, the fact that $ \rho_{(1),\infty}^{(\rm cJ)}( x ;\beta,p,q)$ is real tells us that 
\begin{equation}\label{1.2gz+}     
    \overline{h_{j}(\beta,p,q)} =  
  (-1)^{j+1} h_{j}(\beta,p,q)
   \qquad  \overline{\tilde{h}_{j}(\beta,p,q)} =   (-1)^j \tilde{h}_{j}(\beta,p,q).
   \end{equation}
   It follows from these equations together that $h_j$ alternates pure imaginary then real, while $\tilde{h}_j$ alternates real and then pure imaginary.
  
 Further motivation for our interest in (\ref{1.2e}) is that the case $q=0$  has previously been the subject of some literature for the corresponding two-dimensional generalisation of (\ref{1.1}),
which is up to proportionality specified by the PDF
\begin{equation}\label{1.1x}
\prod_{l=1}^N | z_l |^{\beta Q} e^{- \beta \pi | z_l |^2/2} \prod_{1 \le j < k \le N} | z_k - z_j |^\beta, \qquad z_l  \in \mathbb C.
\end{equation}
The PDF (\ref{1.1x}) has the interpretation as the Boltzmann factor for the two-dimensional one-component plasma, consisting
of $N$ mobile two-dimensional unit charges, repelling via the pair potential $- \log | z - z'|$, with a smeared out uniform neutralising background
of charge density $-1$ in the region $|z| < \sqrt{N/\pi}$; see \cite[\S 4]{BF22a} for a recent review. At the origin is a so-called host charge (or impurity) of strength $Q$.

Let $\rho_{(1),\infty}^{\rm OCP}(\mathbf r; Q)$ denote the density at point $\mathbf r$ in the limit $N \to \infty$ of (\ref{1.1x}). Charge neutrality of the screening cloud
implies
\begin{equation}\label{R1S}
 \int_{\mathbb R^2} ( \rho_{(1),\infty}^{\rm OCP}(\mathbf r; Q)  - 1) \, d \mathbf r = - Q. 
\end{equation}
Two distinct derivations have been given for a sum rule specifying the second moment of the screening cloud \cite{Sa07,JS08} 
 \begin{equation}\label{R2}
  \int_{\mathbb R^2} ( \rho_{(1),\infty}^{\rm OCP}(\mathbf r; Q)  - 1 )  | \mathbf r |^2  \, d \mathbf r =
-    {2 \over \pi \beta} \Big (
(1 - \beta/4) + (\beta/4) Q \Big )  Q .
\end{equation} 
Most recently, under the assumption of analyticity in $Q$, the fourth moment of the screening cloud has been shown to obey \cite{Sa19}
\begin{equation}\label{R3}
 \int_{\mathbb R^2} ( \rho_{(1),\infty}^{\rm OCP}(\mathbf r; Q)  - 1 )  | \mathbf r |^4  \, d \mathbf r =  {1 \over (\pi \beta)^2 }     \Big (
 b_0(\beta) + b_1(\beta) Q + b_2(\beta) Q^2 \Big )  Q,
\end{equation} 
where
 \begin{equation}
  b_0(\beta) = - (\beta - 6) (\beta - 8/3), \qquad    b_1(\beta) =  {2 \over 3} \beta (2 \beta - 7), \qquad    b_2(\beta) =  - {1 \over 3} \beta^2.
\end{equation}  
Earlier \cite{KMST00,CLW15} the polynomial $b_0(\beta)/\beta^2$ appeared as the coefficient of $\mathbf k|^6$ in the expansion of the structure function $S_\infty^{\rm OCP}(\mathbf k;\beta)$.
From the spherical symmetry of $ \rho_{(1)}^{\rm OCP}(\mathbf r; Q)  - 1 $ we have for its Fourier transform
 \begin{equation}\label{1.34}
  \int_{\mathbb R^2} ( \rho_{(1),\infty}^{\rm OCP}(\mathbf r; Q)  - 1) e^{i \mathbf k \cdot \mathbf r  } \,  d \mathbf r  =
  \sum_{j=0}^\infty {(-1)^j \over (j!)^2} \Big ( {|\mathbf k |^2 \over 4} \Big )^j 
   \int_{\mathbb R^2} ( \rho_{(1),\infty}^{\rm OCP}(\mathbf r; Q)  - 1) | \mathbf r |^{2j} \, d \mathbf r,
  \end{equation} 
and so the results (\ref{R1})--(\ref{R3}) give the small $|\mathbf k |$ expansion up and including order $|\mathbf k|^4$, and exhibit a polynomial
structure in $Q$ and $4/\beta$,
 \begin{multline}\label{1.ps}
  \int_{\mathbb R^2} ( \rho_{(1),\infty}^{\rm OCP}(\mathbf r; Q)  - 1) e^{i \mathbf k \cdot \mathbf r  } \,  d \mathbf r  = - Q + {|\mathbf k |^2 \over 2 \pi \beta} \Big (
  (1 - \beta/4) + (\beta/4) Q \Big ) Q  \\
  +    {|\mathbf k |^4 \over 64  (\pi \beta)^2}    \Big (
 b_0(\beta) + b_1(\beta) Q + b_2(\beta) Q^2 \Big )  Q + O(|\mathbf k |^6).
 \end{multline}
 At order $|\mathbf k|^6$ there is evidence that the polynomial structure in $4/\beta$ breaks down, at least for the coefficient of $Q$ as the latter is equivalent to the coefficient of $|\mathbf k|^8$ in the small $|\mathbf k|$ expansion of the structure function $S_\infty^{\rm OCP}(\mathbf k;\beta)$ (\cite{Sa19}; see also Section \ref{S5} below) which is expected to consist of an infinite series in powers of $4/\beta$ 
\cite{KMST00}.
In contrast, in \cite{Sa19} it is argued that the coefficient of $| \mathbf k|^{2j}$ is always a polynomial in $Q$ of degree $j+1$.

Our study of the Fourier transform (\ref{1.2g}) and its small $\tau$ expansion (\ref{1.h}) begins in Section \ref{S2} where we specialise to $\beta = 2$. By doing this use can be made of a third order linear differential equation which then characterises $\rho_{(1),\infty}^{\rm (cJ)}$. Our main finding is a second order difference equation for the odd indexed expansion coefficients of the non-analytic terms in (\ref{1.h}) --- see Proposition \ref{P2.12}.
For $p$ a positive integer the implied sequence terminates for indices great than $2p-1$.
The even indexed expansion coefficients of the non-analytic terms all vanish as does the coefficients of the analytic terms vanish except for $\tilde{h}_0$. The considerations of Section \ref{S2} are repeated in Section \ref{S3}, but now for $\beta = 4$. This is possible due to the derivation of an explicit fifth order linear differential equation characterising $\rho_{(1),\infty}^{\rm (cJ)}$. Proposition \ref{P3.2} then gives a fourth order difference equation satisfied by the expansion coefficients in (\ref{1.h}), which in general are all non-zero. A different approach to studying the expansion coefficients of the Fourier transform is introduced in Section \ref{S4}. This is to derive and apply a hierarchical set of equations referred to as loop equations for the generating function of the Fourier coefficients. In this method $\beta > 0$ is arbitrary. However the computational complexity increases with the order of the expansion coefficient, which restricts the number of terms which can be computed. Five orders are specified in Proposition \ref{P4.2}. A discussion of the explicit functional forms obtained by the loop equation analysis, in the context of known properties of the expansion coefficients of the structure function $S_\infty(\tau;\beta)$, is given in Section \ref{S5}.

\section{The case $\beta = 2$}\label{S2}
\subsection{The confluent hypergeometric kernel and a Bessel kernel}
Denote the normalised form of the PDF proportional to (\ref{1.1}) by $p_{N,\beta}^{(\rm cJ)}(\theta_1,\dots, \theta_N)$. The corresponding $k$-point correlation
function, $\rho_{(N), k}^{(\rm cJ)}(\theta_1,\dots, \theta_k)$ say, is up to normalisation obtained by integrating over all but the first $k$ of the variables in $p_{N,\beta}^{(CJ)}$,
 \begin{equation}\label{2.1}
\rho_{(k), N}^{(\rm cJ)}(\theta_1,\dots, \theta_k) = N (N - 1) \cdots (N - k + 1) \int_{(-\pi, \pi]^{N-k}} d \theta_{k+1} \cdots d \theta_N \,  p_{N,\beta}^{(\rm cJ)}(\theta_1,\dots, \theta_N).
  \end{equation} 
  As relevant for the bulk scaling defined in the paragraph including (\ref{1.2f}), in this we replace $\{\theta_j \}$ in favour of $\{ x_j \}$ and further define
  \begin{multline}\label{2.1a} 
  {\rho}_{(k), \infty}^{(\rm cJ)}(x_1,\dots, x_k;\beta,p,q) \\ =  \lim_{N \to \infty} \Big ( {2 \pi \over N} \Big )^k \rho_{(k), N}^{(\rm cJ)}(-\pi \, {\rm sgn} (x_1) + 2 \pi x_1/N,\dots, -\pi \, {\rm sgn} (x_k) + 2 \pi x_k/N).
 \end{multline}    
  
The case $\beta = 2$ of (\ref{2.1}) and  (\ref{2.1a}) is special. The $k$-point correlation function then has a determinantal structure, with the elements
of the determinant moreover independent of $k$. Specifically, in relation to (\ref{2.1a}) we have
 \begin{equation}\label{2.1a1} 
  {\rho}_{(k), \infty}^{(\rm cJ)}(x_1,\dots, x_k;\beta,p,q) \Big |_{\beta = 2} =  \det \Big [ {K}_\infty^{(p,q)}(x_j, x_l) \Big ]_{j,l = 1,\dots, k},
 \end{equation} 
 where \cite[after the change of variables $1/x_l \mapsto \pi x_l$]{BO01} (see also the introduction of \cite{DKV11} for further references, and  the recent work \cite{FLT21} for an independent derivation)
 \begin{multline}\label{2.1a+} 
 {K}_\infty^{(p,q)}(x, y)     = {(2 \pi)^{2p} | \Gamma(p+1 - i q) |^2 \over \pi  \Gamma(2p+2) \Gamma(2p+1)} e^{- i \pi (x + y) - q \pi({\rm sgn} \,x + {\rm sgn} \, y)/2} { (xy)^{p+1} \over x^2 (x - y)} \\
 \times \Big ( x \, {}_1 F_1(p+1-iq; 2p+2; 2i \pi x)  {}_1 F_1(p-iq; 2p; 2i \pi y)  - (x \leftrightarrow y) \Big ).
 \end{multline}
 Here ${}_1 F_1(a;c;z)$ denotes the confluent hypergeometric function in standard notation. In particular, taking the limit $ y \to x$ using L'H\^{o}pital's rule gives for the bulk scaled density
   \begin{multline}\label{2.1b}
   {\rho}_{(1), \infty}^{(\rm cJ)}(x;\beta,p,q) \Big |_{\beta = 2}  \\=  {(2 \pi)^{2p} | \Gamma(p+1 - i q) |^2 \over \pi \Gamma(2p+2) \Gamma(2p+1)}   e^{-  2 i \pi x  - q \pi{\rm sgn \, x}}  |x|^{2p}  \bigg ( {d \over d x} \Big (  x \, {}_1 F_1(p+1-iq; 2p+2; 2i \pi x) \Big ) \\
   \times {}_1 F_1(p-iq; 2p; 2i \pi x)  -   x \, {}_1 F_1(p+1-iq; 2p+2; 2i \pi x)  {d \over d x}  {}_1 F_1(p-iq; 2p; 2i \pi x)  \bigg ).
     \end{multline}  

     \begin{remark} Use of the Kummer transformation ${}_1 F_1(a,b;z) = e^z {}_1 F_1(b-a,b;z)$ shows that (\ref{2.1b}) is unchanged by the mapping $x \mapsto -x, q \mapsto - q$ as noted in the discussion above (\ref{1.2gz}).
     \end{remark}
 
 \subsection{The case $q = 0$ --- Bessel function asymptotics}
Use of 
 the expression for a particular ${}_1 F_1$ in terms of a Bessel function,
  \begin{equation}\label{2.1c} 
 {}_1 F_1 (p;2p;2iX) = \Gamma(p+1/2) \left( {X \over 2} \right)^{-p+1/2} e^{iX} J_{p-1/2}(X), 
  \end{equation} 
  shows \cite{NS93}
   \begin{equation}\label{2.1aq} 
    {K}_\infty^{(p,q)}(x, y) \Big |_{q=0}  
    = \pi
(xy)^{1/2} { J_{p+1/2}(\pi x)J_{p-1/2}(\pi y)-J_{p-1/2}(\pi x)J_{p+1/2}(\pi y) \over 2(x-y) }.
\end{equation}
Taking the limit $y \to x$ and making use of Bessel function identities then gives \cite[Eq.~(7.49) with $\rho = 1$, $a = p$]{Fo10}
  \begin{equation}\label{2.1aq+} 
 {\rho}_{(1), \infty}^{(\rm cJ)}(x;\beta,p,q) |_{\beta = 2 \atop q=0} =  {\pi^2 |x| \over 2} \Big (
 (J_{p-1/2}(\pi x))^2 +  (J_{p+1/2}(\pi x))^2 - {2 p \over \pi x} J_{p-1/2}(\pi x) J_{p+1/2}(\pi x) \Big ).
 \end{equation}
As a check, it follows from (\ref{2.1aq+}) and trigonometric formulas for the Bessel functions at half integer order that
  \begin{equation}\label{2.1e} 
{\rho}_{(1), \infty}^{(\rm cJ)}(x;\beta,p,q) |_{\beta = 2 \atop p=q=0} = 1.
 \end{equation} 
 We see that this is as required by the rotation invariance of the PDF (\ref{1.1}) for $p=q=0$ and the normalisation implied by bulk scaling.
Setting now $p=1$, shows 
  \begin{equation}\label{2.1f} 
{\rho}_{(1), \infty}^{(\rm cJ)}(x;\beta,p,q) |_{\beta = 2 \atop p=1, q=0} =  1 - \bigg ( {\sin \pi x \over \pi x} \bigg )^2.
 \end{equation}  
 This is in keeping with (\ref{1.3b}) upon recalling the standard fact that in relation to 
 (\ref{1.1}) with $p=q=0$, $\rho_{(2),\infty}^{(\rm cJ)}(x,0;\beta,p,q) |_{\beta = 2}$ is given by the 
 RHS of (\ref{2.1f}); see e.g.~\cite[Eq.~(7.2) with $\rho = 1$]{Fo10}. Generally, substituting any positive integer for $p$ in (\ref{2.1aq+}) reduces
 the Bessel function to a trigonometric form involving $\sin 2 \pi x, \, \cos 2 \pi x$. This is further illustrated by the next simplest example after (\ref{2.1f}), 
 which is to set $p=2$ with the result
  \begin{equation}\label{2.1f+} 
{\rho}_{(1), \infty}^{(\rm cJ)}(x;\beta,p,q) |_{\beta = 2 \atop p=2, q=0} =  1 - {2 \over \pi^2 x^2} - {3 \over 2 \pi^4 x^4}
+ {3 - 2 \pi^2 x^2 \over 2 \pi^4 x^4} \cos 2 \pi x +
{3 \over \pi^3 x^3} \sin 2 \pi x.
 \end{equation}  
 
  In the notation of
 (\ref{1.2g+})  and (\ref{1.h})
  \begin{multline}\label{2.1g}  
  c_\infty^{(\rm cJ)}(\tau;\beta,p,q)  \Big |_{\beta = 2} = \int_{-\infty}^\infty \Big ( {\rho}_{(1), \infty}^{(\rm cJ)}(x;\beta,p,q) |_{\beta = 2} - 1 \Big )
  e^{i \tau x} \, dx  \\ =
  \sum_{j=0}^\infty h_j(\beta,p,q)  |_{\beta = 2}  {\rm sgn} (\tau)  \tau^j +  \sum_{j=0}^\infty \tilde{h}_j(\beta,p,q)  |_{\beta = 2} \tau^j .
 \end{multline} 
 Starting with (\ref{2.1b}), we don't know how to give an explicit special function evaluation of the integral in (\ref{2.1g}). Nonetheless,
 we can use  (\ref{2.1b}), expanded for large $x$,  to deduce the $ h_j(\beta,p,q)$ for $j$ odd in the corresponding small $\tau$ expansion.
 Our ability to do this comes about from Fourier transform theory \cite{Li58}. Thus if
  \begin{equation}\label{2.2a}
 {\rho}_{(1), \infty}^{(\rm cJ)}(x;\beta,p,q) |_{\beta = 2 \atop q=0} - 1 \mathop{\sim}\limits_{x \to \infty} \sum_{n=1}^\infty {c_{2n} \over x^{2n}},
 \end{equation}
 where only non-oscillatory terms are recorded on the RHS,
 then
   \begin{equation}\label{2.2b}   
 \int_{-\infty}^\infty \Big ( {\rho}_{(1), \infty}^{(\rm cJ)}(x;\beta,p,q) |_{\beta = 2  \atop q=0} - 1 \Big )
  e^{i \tau x} \, dx     \mathop{\sim}\limits_{\tau \to 0}  \pi \sum_{n=1}^\infty {(-1)^n c_{2n} \over (2n-1)!} | \tau |^{2n-1}.
  \end{equation}
 Generally this approach can access the terms in (\ref{2.1g}) singular in $\tau$ only. The analytic terms  
are not accessible via this method. 

According to  (\ref{2.1b}), the expansion  (\ref{2.2a}) as applies to the case
 $q=0$ can be deduced from
knowledge of the large $x$ asymptotic expansion of the Bessel function. The latter is given by
  \cite[Eq.~(10.17.3)]{DLMF}
  \begin{equation}\label{2.2d}    
 J_\nu(z) \sim \Big ( {2 \over \pi z} \Big )^{1/2} \bigg (
 \cos \omega \sum_{k=0}^\infty (-1)^k {a_{2k}(\nu) \over z^{2k}} - \sin \omega \sum_{k=0}^\infty (-1)^k
  {a_{2k+1}(\nu) \over z^{2k+1}}  \bigg ),
  \end{equation}
  where with  $(u)_s := u(u-1) \cdots (u+s-1)$
    \begin{equation}\label{2.2e}    
  \omega = z - {1 \over 2} \nu \pi - {1 \over 4} \pi, \quad a_k(\nu) = {1 \over (-2)^k k!} \Big ( {1 \over 2} - \nu \Big )_k
 \Big ( {1 \over 2} + \nu \Big )_k. 
  \end{equation}

  \begin{proposition}
  Consider the large $x$ asymptotic expansion of $ {\rho}_{(1), \infty}^{(\rm cJ)}(x;\beta,p,q) |_{\beta = 2, q=0}$ restricted to non-oscillatory terms. This is given by (\ref{2.2a}) with $c_{2n} = c_{2n}(p)$, where
  \begin{multline}\label{ca1}
  c_{2n}(p)  ={ (-1)^n \over \pi^{2n} } \bigg ( {1 \over 2} \sum_{s=0}^n \Big ( a_{2s}(p-1/2)  a_{2(n-s)}(p-1/2) + a_{2s}(p+1/2)  a_{2(n-s)}(p+1/2) \Big ) \\
  -  {1 \over 2}  \sum_{s=0}^{n-1} \Big ( a_{2s+1}(p-1/2)  a_{2(n-s)-1}(p-1/2) + a_{2s+1}(p+1/2)  a_{2(n-s)-1}(p+1/2) \Big ) \\
  + p   \sum_{s=0}^{n-1} \Big ( a_{2s}(p-1/2)  a_{2(n-s)-1}(p+1/2) -a_{2s+1}(p-1/2)  a_{2(n-s)-2}(p+1/2) \Big ) \bigg ),
  \end{multline}
  with $\{a_k(\nu)\}$ as in (\ref{2.2e}).
  \end{proposition}
  
\begin{proof}
We substitute (\ref{2.2d}) in   (\ref{2.1aq+}). Applying simple trigonometric identities allows the non-oscillatory terms
to be separated from the oscillatory terms. Finally, the coefficient of $1/x^{2n}$ is extracted using the formula for multiplication
of power series.
\end{proof} 

\begin{remark}
The use of computer algebra to compute (\ref{ca1}) for small values of $n$ and general $p$ suggests the simplified form
  \begin{equation}\label{2.2f}  
   c_{2n}(p)  = - { \alpha_n  \over \pi^{2n} }  \prod_{l=0}^{n-1} (p^2 - l^2),
  \end{equation} 
  where $\alpha_n$ is a rational number. Regarding the latter, the computer algebra computation gives
  \begin{equation}\label{2.2g}    
  \alpha_1 = {1 \over 2}, \quad \alpha_2 = {1 \over 8}, \quad \alpha_3 = {1 \over 16}, \quad \alpha_4 = {5 \over 128}, \quad \alpha_5 = {7 \over 256}, \quad  
  \alpha_6 = {21 \over 1024}.
   \end{equation} 
   We note in particular that (\ref{2.2f}) implies $c_{2n}(p) = 0$ for $n > p$, and that $c_{2n}(p)$ is a polynomial in $p^2$ of degree $n-1$. In fact this simplified form,
   together with the explicit value of $\alpha_n$, can be obtained by taking a different approach to the large $x$ asymptotic expansion of
$ {\rho}_{(1), \infty}^{(\rm cJ)}(x;\beta,p,q) |_{q=0,\beta = 2}$, namely via a differential equation, which is to be done next.
   \end{remark} 

 \subsection{The case $q = 0$ --- a third order linear differential equation}
 In the work \cite[\S 3.1.2]{FR21} the scaled density $(1/\pi)\rho_{(1),\infty}^{\rm cJ}(x/\pi;\beta,p,q) |_{\beta = 2, q=0}$ was shown to satisfy the
 third order linear differential equation
   \begin{equation}\label{R}  
   x^2 R'''(x) + 4 x R''(x) + (2 - 4p^2 + 4x^2) R'(x) - {4p^2 \over x} R(x) = 0.
   \end{equation} 
   This can be used to provide an alternative approach to the computation of the coefficients in the expansion (\ref{2.2a}).

  \begin{proposition}\label{P2.3}
  The expansion  (\ref{2.2a}) holds with
   \begin{equation}\label{R1}   
  c_{2} = - {1 \over 2 \pi^2} p^2, \qquad c_{2n} = - {1 \over \pi^{2n}}  {(2n-3)!! \over (2n)!!} \prod_{l=0}^{n-1} (p^2 - l^2), \quad (n \ge 2).
  \end{equation}  
  \end{proposition}
  
\begin{proof}
Due to the scaling of $x$ in the density assumed in the derivation of (\ref{R}), the expansion  (\ref{2.2a}) transforms to
 \begin{equation}\label{2.2aY} 
 R(x) - {1 \over \pi}  \mathop{\sim}\limits_{x \to \infty}  {1 \over \pi}  \sum_{n=1}^\infty {\pi^{2n} c_{2n} \over x^{2n}}.
  \end{equation}    
Substituting   (\ref{2.2aY}) in (\ref{R}) and equating like (negative) powers of $x$ implies
 \begin{equation}\label{2.2aZ} 
  c_{2} = - {1 \over 2 \pi^2} p^2, \qquad c_{2n+2} = {1 \over \pi^2} {2n - 1 \over 2 (n+1)} (p^2 - n^2) c_{2n}, \quad (n \ge 1).
    \end{equation}    
  Iterating the first order recurrence (\ref{R1}) follows.
  \end{proof}
  
  \begin{remark} 
  In the notation of (\ref{2.2f}), (\ref{R1}) gives $\alpha_1 = 1/2$ and
  $$
  \alpha_n =   {(2n-3)!! \over (2n)!!}
  $$ 
  for $n \ge 2$. This indeed reproduces the
  values obtained in (\ref{2.2g}).
  \end{remark}
  
  With $r(x) : = R(x) - 1/\pi = {\rm O}(x^{-2})$ as implied by (\ref{2.2aY}), there is a well defined meaning to taking the Fourier transform of
  (\ref{R}). Upon integration by parts and an additional differentiation with respect to $\tau$ this reads
     \begin{equation}\label{R3y}  
     - {d^3 \over d \tau^3} \Big ( (\tau^3 - 4 \tau) \hat{r}(\tau) \Big ) + 4 {d^2 \over d \tau^2} (\tau^2 \hat{r}(\tau)) - 2 (1 - 2 p^2) {d \over d \tau} ( \tau \hat{r}(\tau)) -
     4 p^2 \hat{r}(\tau) = 0,
     \end{equation}  
     where
     $$
      \hat{r}(\tau)  := \int_{-\infty}^\infty r(x) e^{i \tau x} \, dx.
      $$
      
    \begin{proposition}\label{P2.4}
 Let 
    \begin{equation}\label{Rb0}
    b_0 = \int_{-\infty}^\infty \Big ( R(x) - {1 \over \pi} \Big ) \, dx, \qquad b_1 = {p^2 \over 2 \pi}
  \end{equation} 
  and specify $\{ c_{2n}\}$ by Proposition \ref{P2.3}.   
 The small $\tau$ expansion   
   \begin{equation}\label{R4}
   \hat{r}(\tau) =     \sum_{n = 0}^\infty b_n  ( { \pi \tau }  )^n  , \qquad \tau > 0, 
     \end{equation}        
 substituted in (\ref{R3y}) 
has the unique solution
  \begin{equation}\label{R5}
  b_{2n - 1} =   \pi  {(-1)^n c_{2n} \over (2n - 1)!} \quad (n \ge 1), \qquad b_{2n} = 0 \quad (n \ge 1).
  \end{equation}      
 \end{proposition}
 
\begin{proof}
Substituting (\ref{R4}) in (\ref{R3y}) gives
$$
b_{n+2} = - { n \over \pi^2 (n+3)(n+2)(n+1) } \Big ( p^2 - {(n+1)^2 \over 4} \Big ) b_n \quad (n \ge 0).
$$
The value of $b_0$ is not restricted by the differential equation, nor is the value of $b_1$, while for $n$ even and positive the recurrence gives $b_n = 0$. For odd subscripts,
replacing $n$ by $2n-1$ allows the recurrence to be related to that in (\ref{2.2aZ}), which implies 
(\ref{R5}).
\end{proof}

\begin{remark} $ $ \\
1.~The implied coefficient of the third derivative in (\ref{R3y}) has a zero at $\tau = 0$ and $\tau = \pm 2 $, telling us that the radius of
convergence of (\ref{R4}) is $2$. \\
2.~Since $r(x)$ as relates to (\ref{R3y}) is an even function of $x$, $\hat{r}(\tau)$ must be an even function of $\tau$. Hence for $\tau < 0$
(\ref{R4}) is now a power series in $|\tau|$, as is consistent with (\ref{2.1g}). In particular the coefficient of $\tau$ must then relate
to the coefficient of $1/x^2$ in the large $x$ expansion of $r(x)$ as implied by the relation between (\ref{2.2a}) and (\ref{2.2b}). This justifies
the choice of $b_1$ in (\ref{Rb0}) which as already mentioned otherwise is not determined by (\ref{R3y}). \\
3.~The log-gas interpretation of the factor $| 1 + e^{i \theta} |^{\beta p}$ in the circular Jacobi weight $w^{(\rm cJ})(\theta)$ of (\ref{1.1}),
otherwise interpreted as a spectrum singularity of degeneracy $p$ at $\theta = \pi$, is as a fixed charge of strength $p$. From
this log-gas picture, perfect screening of the fixed charge (see e.g.~\cite[\S 14.1]{Fo10}) implies
   \begin{equation}\label{Rb01}
     \int_{-\infty}^\infty \Big ( R(x) - {1 \over \pi} \Big ) \, dx = - p
  \end{equation}    
and thus from (\ref{Rb0}) that $b_0 = - p$.
\end{remark}

\begin{corollary}
Specify $\{ c_{2j}\}$ by Proposition \ref{P2.3} supplemented by $c_0 = 0$. For $\beta = 2$ and $q=0$, the expansion (\ref{1.h}) holds true with
  \begin{equation}\label{h2a}
     h_{2j}(\beta, p , q) \Big |_{\beta = 2 \atop q=0} = c_{2j} \:\: (j \ge 0), \qquad 
 h_{2j-1}(\beta, p , q) \Big |_{\beta = 2 \atop q=0} = \tilde{h}_{j}(\beta, p , q) \Big |_{\beta = 2 \atop q=0} = 0  \:\: (j \ge 1) .
  \end{equation} 
  The perfect screening sum rule (\ref{Rb01}) gives $ \tilde{h}_{0}(\beta, p , q) |_{\beta = 2 , p=0} = - p$. 
  \end{corollary}

 \subsection{Third order linear differential equation for the case $q \ne 0$}
 A third order linear differential equation satisfied by the scaled density 
 $(1/\pi)\rho_{(1),\infty}^{\rm cJ}(x/\pi;\beta,p,q) |_{\beta = 2}$ for general $p,q$
 can be obtained by following the same procedure as used in \cite{FR21}
 to deduce (\ref{R}), suitably generalised to involve the extra parameter $q$.
 
 \begin{proposition}\label{P2.9}
 The scaled density $(1/\pi)\rho_{(1),\infty}^{\rm cJ}(x/\pi;\beta,p,q) |_{\beta = 2}$ obeys
 the third order linear differential equation 
  \begin{equation}\label{Rq}  
   x^3 R'''(x) + 4 x^2 R''(x) + 2x (1 - 2p^2 -4qx +  2x^2) R'(x) - 4( p^2 + q x) R(x) = 0.
   \end{equation}
  \end{proposition}

\begin{proof}
For the Cauchy ensemble with weight (\ref{1.2b}) and $\beta = 2$
 let $\rho_{(1),N}^{(Cy)}(t)$ denote the corresponding density. 
With $g(t) = (1 + t^2) \rho_{(1),N}^{(Cy)}(t)$
we know
from \cite[Eq.~(3.33) with $r(t) = g(t)$, $\alpha_1 = p$, $\alpha_2 = -q$]{FR21} that
\begin{multline}\label{3.33}
(1 + t^2 )^2 g''' + 2 t (1 + t^2) g'' + 4 \Big ( (p^2 t - (N + p) q)(g - t g') \\
+ ((N + p)^2 + (N + p) q t - p^2 -  q^2) g' \Big ) = 0.
\end{multline}

From the theory of the second paragraph of the Introduction it follows $2 g(t) |_{t = \tan (\theta/2)}=
\rho_{(1),N}^{(\rm cJ)}(\theta)$. We want to scale $\theta$ about the spectrum singularity at $\theta = \pi$,
and as a normalisation we require that at large distances in the scaled variable the density is $1/ \pi$. These
requirements are met by choosing $\theta = \pi + 2X/N$. With $t = \tan (\theta/2)$ and $N$ large this implies
$t \sim - N/ X$. We make this substitution in (\ref{3.33}). Equating terms at leading order in $N$ in the resulting
equation gives (\ref{Rq}).
\end{proof}

For $q \ne 0$, it follows from the functional form (\ref{1.1}) that $ {\rho}_{(1), \infty}^{(\rm cJ)}(x;\beta,p,q)$ is not
an even function of $x$, in distinction to the case $q=0$. In place of (\ref{2.2a}) we now expect
 \begin{equation}\label{2.2aq}
 {\rho}_{(1), \infty}^{(\rm cJ)}(x;\beta,p,q) |_{\beta = 2 } - 1 \mathop{\sim}\limits_{x \to \pm \infty} \sum_{n=1}^\infty {d_{n} \over x^{n}},
 \end{equation}
 for the asymptotic form of the non-oscillatory terms.  The analogue of (\ref{2.2b}) for the non-analytic terms in the small $\tau$ form of the
 Fourier transform is then \cite{Li58}
  \begin{equation}\label{2.2bq}   
 \int_{-\infty}^\infty \Big ( {\rho}_{(1), \infty}^{(\rm cJ)}(x;\beta,p,q) |_{\beta = 2  } - 1 \Big )
  e^{i \tau x} \, dx     \mathop{\sim}\limits_{\tau \to 0}  \pi \sum_{n=1}^\infty d_n {i^n \over (n-1)!} \tau^{n-1} {\rm sgn} \, \tau.
  \end{equation}
  As for the result of Proposition \ref{P2.3}, the differential equation (\ref{Rq}) can be used to
  determine the coefficients $d_n$ in (\ref{2.2aq}).

 \begin{proposition}\label{P2.3q}
  The expansion  (\ref{2.2aq}) holds with
   \begin{equation}\label{R1q}   
 d_1 = - {q \over \pi}, \quad d_2 = - {p^2 + q^2 \over 2 \pi^2} , \quad d_{n+2} = {1 \over \pi} {2n + 1 \over n + 2} q d_{n+1} +  {1 \over \pi^2} {n - 1 \over n + 2} \Big (
 p^2 - {n^2 \over 4} \Big ) d_n \: \: (n \ge 1).
  \end{equation}  
  \end{proposition}
  
\begin{proof}
As with (\ref{2.2aY}), to relate the expansion (\ref{2.2aq}) to an expansion of $R(x)$ in (\ref{Rq}) requires a scaling 
$x \mapsto x/\pi$ to give
 \begin{equation}\label{2.2aY+} 
 R(x) - {1 \over \pi}  \mathop{\sim}\limits_{x \to \infty}  {1 \over \pi}  \sum_{n=1}^\infty {\pi^{n} d_{n} \over x^{n}}.
  \end{equation}  
  The result now follows by  
substituting   (\ref{2.2aY+}) in (\ref{Rq}) and equating like (negative) powers of $x$.
  \end{proof}

  We note that setting $q=0$ in (\ref{R1q}) reclaims the result of (\ref{R1}). However unlike the situation with $q=0$ the recurrence in
  (\ref{R1q}) does not admit a simple functional form for its solution. Structural points of interest are that $d_{2n}$ is equal to $(p^2+q^2)$ times
  a polynomial in $p^2$ and $q^2$ of degree $n-1$, while for $n \ge 2$, $d_{2n-1}$ is equal to $q(p^2+q^2)$ times a polynomial in
  $p^2$ and $q^2$ of degree $n-2$. Specifically, as some low order examples
  \begin{equation}\label{2.2T}  
  d_3 = - { q (p^2 + q^2)\over 2 \pi^3}, \: \: d_4 =  {(p^2+q^2)(1-p^2-5q^2) \over 8 \pi^4}, \: \:
  d_5 =  { q(p^2+q^2)(5-3p^2-7 q^2) \over 8 \pi^5}.
   \end{equation}  
   And as we know from (\ref{2.2bq}),
  knowledge of $\{ d_n \}$ 
  gives the explicit form of the coefficients $\{ h_n(\beta,p,q)|_{\beta = 2}  \}$ in the functional form (\ref{1.h}) according to
  \begin{equation}\label{2.2y} 
   h_{n-1}(\beta,p,q)|_{\beta = 2}  = \pi d_n {i^n \over (n-1)!}, \quad n \ge 1.
  \end{equation}   
  
  \begin{remark} $ $ \\
  1.~Starting from (\ref{2.1b}), and using the appropriate generalisation of (\ref{2.2d}) as can be found in \cite{DLMF}, we have independently
  verified that with respect to non-oscillatory terms
  \begin{equation}\label{nv}
  R(x) - {1 \over \pi} \mathop{\sim}\limits_{x \to \infty} - {q \over \pi x},
  \end{equation}
  as is consistent with Proposition \ref{P2.3q}. \\
  2.~Assuming the validity of   (\ref{Rb01}) for $q \ne 0$, the result (\ref{R1q}) for $d_1$ used in (\ref{2.2bq}) implies
  \begin{equation}\label{Rb01q}
   \lim_{\tau \to 0^\pm}   \int_{-\infty}^\infty \Big ( R(x) - {1 \over \pi} \Big ) e^{i \tau x} \, dx = - p - i q \, {\rm sgn} \, \tau,
  \end{equation}    
telling us that for $q \ne 0$ the Fourier transform is a discontinuous function of $\tau$, and that in (\ref{1.h}), $\tilde{h}_0(\beta,p,q)|_{\beta = 2} = -p$. \\
3.~Difference equations satisfied by moments of the spectral density for classical ensembles is a rich theme in random matrix theory; references include \cite{HZ86,Le04,Le09,MS13,WF14,CMOS19,CCO20,ABGS20,GGR21,FK22,FLSY23}. A feature of these settings is a global scaling for which the spectral density has  compact support. However in the present setting of the non-compactly supported bulk scaled spectral density about a spectrum singularity, there is no literal meaning of the moments due to the slow decay at infinity. Hence the non-analytic terms in the expansion (\ref{1.h}) of the Fourier series. Note that this is in contrast to the two-dimensional case as exhibited in (\ref{1.34}), even though the support there is also non-compactly supported.
\end{remark}

It remains to investigate the coefficients $\{\tilde{h}_j(\beta,p,q)|_{\beta = 2}\}$ in the expansion (\ref{1.h}) for $j \ge 1$. For this purpose we
introduce $ r(x) = R(x) - 1/ \pi$ and proceed to take the Fourier transform of (\ref{Rq}) to deduce as the $q \ne 0$ generalisation of (\ref{R3y})
 \begin{multline}\label{R3yq}  
     - {d^3 \over d \tau^3} \Big ( (\tau^3 - 4 \tau) \hat{r}(\tau) \Big ) +  {d^2 \over d \tau^2} \Big ((
     4\tau^2  + 8iq \tau) \hat{r}(\tau)   \Big )  - {d \over d \tau} \Big ( (   2 (1 - 2 p^2) \tau   + 4 i q  )
      \hat{r}(\tau) \Big ) \\
     -4 p^2 \hat{r}(\tau) = 0.
     \end{multline} 
 The following result now follows by direct substitution.
  
   \begin{proposition} \label{P2.12}
 The small $\tau$ expansion   
   \begin{equation}\label{R4e}
   \hat{r}(\tau) =     \sum_{n = 0}^\infty e_n  \Big ( {\tau \over \pi} \Big )^n  , \qquad \tau > 0, 
     \end{equation}        
 substituted in (\ref{R3yq}) implies the recurrence for $\{ e_n \}$
    \begin{equation}\label{R5e}
    (n+3)(n+2)(n+1) e_{n+2} + i q (2n+3)(n+1) e_{n+1} + n \Big ( p^2 - {(n+1)^2 \over 4} \Big ) e_n = 0, \quad (n \ge 0).
     \end{equation}        
 \end{proposition}
 
 The values of $e_2, e_3, \dots$ implied by (\ref{R5e}) are independent of $e_0$ but depend on $e_1$. However in keeping with the
 circumstances of Proposition \ref{P2.4} the value of $e_1$ is not determined by the differential equation (\ref{R3yq}). We note that writing
 $e_n = \tilde{d}_{n+1} i^{n+1}/n!$ in (\ref{R5e}) shows $\{ \tilde{d}_n \}$ satisfies the same recurrence relation as $\{ d_n \}$ in
 Proposition \ref{P2.3q}. In the context of the expansion (\ref{1.h}) this implies that for $\beta = 2$ the sequences
 $\{ h_j(\beta,p,q) \}$ and $\{ \tilde{h}_j(\beta,p,q) \}$ satisfy the same recurrence.  While all terms in the sequence
 $\{ h_j(\beta,p,q)|_{\beta = 2} \}$  are known according to Proposition \ref{P2.3q} and (\ref{2.2y}), the sequence members
 $ \tilde{h}_j(\beta,p,q)|_{\beta = 2} $ for $q \ne 0$ and $j \ge 2$
  depend on the
 value of $ \tilde{h}_1(\beta,p,q)|_{\beta = 2} $ for $q \ne 0$, which  is not known from the above considerations. However results from Section \ref{S4} --- specifically by taking the imaginary part of $\alpha_1 |_{\beta = 2}$ in (\ref{iS4}) ---  
 allow us to deduce that we have $\tilde{h}_1(\beta,p,q)|_{\beta = 2} = 0$ for general $q $ and hence $\tilde{h}_j(\beta,p,q)|_{\beta = 2} =0$ for $j \ge 1$.

\section{The case $\beta = 4$}\label{S3}
The scaled statistical state in the neighbourhood of the spectrum singularity for $\beta = 4$ is a Pfaffian point process, rather than the simpler
determinantal point process  for $\beta = 2$ \cite{FN01,FLT21}. Perhaps surprisingly then, the functional form of the of the $\beta = 4$ density
is functionally related to the $\beta =2$ density. This is simplest to state with $q=0$, for which  from \cite[Eq.~(3.34) with $a=p$ and $X=Y$]{FN01} reads
 \begin{equation}\label{FN}
 \rho_{(1),\infty}^{(\rm cJ)}(x;\beta,p,q) \Big |_{\beta = 4 \atop q = 0} =
 \rho_{(1),\infty}^{(\rm cJ)}(2x;\beta,p,q) \Big |_{\beta = 2 \atop q = 0, p \mapsto 2p}  - \pi p {J_{2p-1/2}(2x) \over x^{1/2}}
 \int_0^x s^{-1/2} J_{2p+1/2}(2s) \, ds,
  \end{equation}
  where for convenience it is assumed $x>0$ (as previously remarked, for $q=0$ the density is an even function).
An analogous result for $q \ne 0$ is given in \cite{FLT21}. However for our interest in the expansion (\ref{1.h}) we will not make use of such an explicit
expression. Rather the key feature of the $\beta = 4$ density for our purposes is that it, like the $\beta =2$ density, satisfies a linear differential equation albeit now of degree 5.

\begin{proposition}\label{P3.1}
Let $\tilde{p} = p - 2p^2$. The scaled density $ (1 / \pi) \rho_{(1), \infty}^{\mathrm{cJ}}(x / \pi ; \beta, p, q)|_{\beta = 4} $ satisfies the differential equation 
    \begin{align*}
	x^5R^{(5)}(x)&+10x^4R^{(4)}(x)+x^3(20x^2-20qx+22+10\tilde{p})R{'''}(x)\\
+&x^2(64x^2-76qx+4+44\tilde{p})R{''}(x)\\
	+&4x(16x^4-32qx^3+4(4q^2+4\tilde{p}+1)x^2-q(6+16\tilde{p})x+4\tilde{p}^2+7\tilde{p}-1)R{'}(x)\\
	+&8(-4qx^3+4(q^2+\tilde{p})x^2-q(6\tilde{p}-1)x+2\tilde{p}^2)R(x)=0.
\end{align*}
\end{proposition}

\begin{proof}
A fifth order linear differential equation for the $\beta = 4$ Jacobi ensemble density has been given in \cite[Th.~2]{RF21}. Making use of the relation between the Jacobi and Cauchy averages (\ref{1.2d}) we can deduce from this a fifth order differential equation for the quantity $(1+t^2) \rho_{(1),N}^{(Cy)}(t)$, where $\rho_{(1),N}^{(Cy)}(t)$ denotes the $\beta = 4$ Cauchy ensemble density corresponding to the weight (\ref{1.2b}). As in the proof of Proposition \ref{P2.9}, the computation is concluded by substituting $t=-N/X$ as corresponds to a hard edge scaling and equating to leading order in $N$. Due to the large number of terms involved, the required steps were all carried out using computer algebra.
\end{proof}

As for the differential equation (\ref{Rq}), the differential equation of Proposition \ref{P3.1} has a unique solution of the form
(\ref{2.2aY+}).

\begin{proposition}\label{P3.2}
Substituting the expansion (\ref{2.2aY+}) in the differential equation of Proposition \ref{P3.1}, but with the coefficients renamed from $\{d_n\}$ to $\{g_n\}$ for distinction, one obtains that
\begin{align*}
	g_1=-\frac{q}{2\pi},\quad g_2=\frac{1}{8\pi^2}(2\tilde{p}-q^2),\quad g_3=-\frac{1}{16\pi^3} q\left(-1-2 \tilde{p}+q^2\right)\\
	g_4=\frac{1}{128\pi^4}\left(-16 \tilde{p}-4 \tilde{p}^2+19 q^2+12 \tilde{p} q^2-5 q^4\right),
\end{align*}
with the higher order coefficients $g_5,g_6,\dots$ then determined by the fourth order recurrence relation
\begin{align*}
	64(n+&4)\pi^{4}g_{n+4}-32q(11+4n)\pi^{3}g_{n+3}\\
	+&4\left(32+54 n+29 n^2+5 n^3+8 \tilde{p}(3+2 n)+24 q^2+16 n q^2\right)\pi^{2}g_{n+2}\\
	-&4q\left(4 \tilde{p}(1+4 n)+n\left(4+11 n+5 n^2\right)\right) \pi g_{n+1}\\
	+&\left((n-1)\left(16 \tilde{p}^2+2 \tilde{p} n(5 n-2)+n^2\left(-2+n+n^2\right)\right)\right)g_n=0.
\end{align*}
\end{proposition}

As with going from (\ref{2.2aq}) to (\ref{2.2bq}), then to (\ref{2.2y}),
knowledge of $\{g_n\}$ tells us that
\begin{equation}\label{2.2y4} 
   h_{n-1}(\beta,p,q)|_{\beta = 4}  = \pi g_n {i^n \over (n-1)!}.
  \end{equation}  

\begin{remark} $ $ \\
1.~A fifth order linear differential equation is also known for the density of the Jacobi ensemble in the case $\beta = 1$
\cite{RF21}. However this is not independent of the corresponding differential equation in the case $\beta = 4$. Writing $g_n = g_n(\beta,p,q)$, where $\{g_n\}$ is used in the same sense as Proposition \ref{P3.2}, the dependence implies
\begin{equation}\label{2.2z4}
g_n(\beta,p,q) |_{\beta = 1} = (-2)^{n+1}
g_n(\beta,-p/2,-2q) |_{\beta = 4}. 
\end{equation}
This functional equation is generalised in (\ref{4.30}) below.
\\
2.~The fifth order differential equation of Proposition \ref{P3.1} can be transformed into a fifth order differential equation for the Fourier transform of $(R(x) - 1/\pi)$. Now substituting the small $\tau$ expansion (\ref{R4e}) gives a fifth order recurrence for the coefficients $\{e_n\}$. Upon the substitution $e_n = \tilde{g}_{n+1} i^{n+1}/n!$, this recurrence becomes that of Proposition \ref{P3.2} with $\{g_n\}$ replaced by $\{\tilde{g}_n\}$. As in the discussion below Proposition \ref{P2.12} in the case $\beta = 2$, this implies that for $\beta = 4$ the sequences
 $\{ h_j(\beta,p,q) \}$ and $\{ \tilde{h}_j(\beta,p,q) \}$ of the expansion (\ref{1.h}) satisfy the same recurrence. According to (\ref{2.2y4}) and Proposition \ref{P3.2} the first of these is fully determined. However the latter requires specification of $\tilde{g}_2, \tilde{g}_3, \tilde{g}_4$ as initial conditions, which are not determined by the differential equation. However, we will see that these initial conditions can be accessed using result deduced from results of the next section (specifically (\ref{4.27}) and (\ref{4.28})), which imply
 \begin{equation}
 \tilde{g}_2 = - {q \over \pi^2}, \quad
 \tilde{g}_3 = {1 \over 8 \pi^3} (- \tilde{p} + q^2), \quad
 \tilde{g}_3 = {1 \over 16 \pi^4}q ( - 1 - 3 \tilde{p} + 2 q^2).
 \end{equation}
\end{remark}

\section{A loop equation approach}\label{S4}
\subsection{Connected correlators}
Consider the Stieltjes transform
 \begin{equation}\label{FNa}
 \overline{W}_1(x;N,\beta,p,q)  : = \int_{0}^{2\pi} {\rho_{(1),N}^{\widetilde{\rm (cJ)}} ( \theta;\beta,p,q) \over x - e^{i \theta}} \, d \theta, \qquad x \notin {\mathcal C}_1,
   \end{equation}     
where ${\mathcal C}_1$ denotes the unit circle in the complex plane. Expanding for large $x$ shows that this quantity relates to the Fourier components
of the density as specified in (\ref{1.2e}),
 \begin{equation}\label{FNb}  
  \overline{W}_1(x;N,\beta,p,q)  = {1 \over x} \sum_{k=0}^\infty {c_{k}^{(\widetilde{\rm cJ)}}(N,\beta,p,q) \over x^k}, \quad |x| > 1.
  \end{equation}  
  
  A primary aim of this paper is to obtain the small $\tau$ expansion of the Fourier transform $c_\infty^{\rm (cJ)}(\tau;\beta,p,q)$ as specified in
  (\ref{1.2g+}). To see how $ \overline{W}_1(x;N,\beta,p,q)$ relates to this aim, suppose that this quantity admits a $1/N$ expansion
 \begin{equation}\label{E2}
  \overline{W}_1(x;N,\beta,p,q)  = {N\over x}   +  N \sum_{l=1}^\infty  { W_1^{l}(x;\beta,p,q) \over N^l}.
  \end{equation}
  Suppose furthermore that each $ W_1^{l}(x;\beta,p,q)$ when expanded for large $x$ analogous to (\ref{FNb})
  has coefficients which are polynomials in   $k$ of degree $l-1$. Denote the leading coefficient in the polynomial by $ \alpha_{l-1}(\beta,p,q)$.
  It then follows from the definition (\ref{1.2h}) that
  \begin{equation}\label{E2a} 
c_\infty^{\widetilde{\rm (cJ)}}(\tau;\beta,p,q) = \sum_{l=0}^\infty \Big ( {\tau \over 2 \pi} \Big )^l \alpha_l(\beta,p,q).
   \end{equation}
   In fact this general approach has been used to determine $\{   \alpha_l(\beta,p,q) |_{p=1, q= 0} \}$ up to and including $l=10$ in \cite{WF15}, thereby
   reclaiming the explicit form of the polynomials $\{ p_j(y) \}$ in (\ref{1.17a}) for $j=1,\dots,9$, first obtained in \cite{FJM00}.
  
  In the loop equation formalism, the computation of the coefficients in the series of (\ref{E2}) up to the $j$-th requires knowledge of the
  large $N$ expansion of particular truncations of the multipoint correlators
  \begin{equation}\label{E2b}   
  W_n(x_1,\dots,x_n;N,\beta,p,q) := \Big \langle G(x_1) \cdots G(x_n) \Big \rangle, \qquad G(x) := \sum_{j=1}^N {1 \over x - e^{i \theta_j}}.
 \end{equation}  
 These truncations, denoted $\overline{W}_n$, have for $n=2$  the simple covariance form
   \begin{equation}\label{E2c}   
   \overline{W}_2(x_1, x_2)  =  \Big \langle  (  G(x_1) - \langle G(x_1) \rangle  )    (  G(x_2) - \langle G(x_2) \rangle  ) \Big \rangle,
 \end{equation}      
and similarly for $n=3$. For general $n \ge 2$ the truncations, referred to as connected correlators, are 
symmetric linear combinations of $\{   W_m(x_{j_1},\dots,x_{j_m};N,\beta,p,q) \}_{m=1,\dots,n}$ with
$1 \le j_1 < \cdots < j_m \le n$ defined so that they are
generating functions for the mixed cumulants; see e.g.~\cite[Eq.~(1.17)]{Fo22}.  A key feature of the connected correlators
is that their order of  decay in  a $1/N$ expansion increases as  $n$ increases,
   \begin{equation}\label{E2d}  
  \overline{W}_n(x_1,\dots,x_n;N,\beta,p,q)  = N^{2-n}    \sum_{l=0}^\infty  { \overline{W}_n^{l}(x;\beta,p,q) \over N^l};  
 \end{equation}    
see \cite{BG11,WF15}. 
To compute
$ \overline{W}_1^{l}(x;\beta,p,q) $ in (\ref{E2a}) up to order $l=j$,
we will require knowledge of the expansion (\ref{E2d}) of $  \overline{W}_n$ for $n=2,\dots,j+1$ up to order $l=n-j-1$.

\subsection{The loop equation hierarchy}
The connected correlators have  previously been analysed in the case of the Jacobi $\beta$ ensemble based
on (\ref{1.2}) with the weight (\ref{1.2c}). They satisfy a hierarchy of equations --- referred to as loop equations --- specified by
 \cite[Eq.~(4.6)]{FRW17}
  \begin{align}\label{4.5}
0&=\left((\kappa-1)\frac{\partial}{\partial x_1}+\left(\frac{\lambda_1}{x_1}-\frac{\lambda_2}{1-x_1}\right)\right)\overline{W}_n^{\rm J}(x_1,J_n) 
- \frac{n-1}{x_1(1 - x_1)} \overline{W}_{n-1}^{\rm J}(J_n) 
\\&\quad  + { \chi_{n = 1} \over x_1 ( 1 - x_1)} \Big ( (\lambda_1 + \lambda_2 + 1)N + \kappa N (N - 1) \Big )-\frac{\chi_{n \ne 1}}{x_1(1-x_1)}\sum_{k=2}^nx_k\frac{\partial}{\partial x_k}\overline{W}_{n-1}^{\rm J}(J_n)\nonumber
\\&\quad
+ \chi_{n \ne 1} \sum_{k=2}^n\frac{\partial}{\partial x_k}\left\{\frac{\overline{W}_{n-1}^{\rm J}(x_1,\ldots,\hat{x}_k,\ldots,x_n)-\overline{W}_{n-1}^{\rm J}(J_n)}{x_1-x_k}+\frac{1}{x_1}\overline{W}_{n-1}^{\rm J}(J_n)\right\}\nonumber
\\\nonumber &\quad+\kappa\left[\overline{W}_{n+1}^{\rm J}(x_1,x_1,J_n)+\sum_{J\subseteq J_n}\overline{W}_{|J|+1}^{\rm J}(x_1,J)\overline{W}_{n-|J|}^{\rm J}(x_1,J_n\setminus J)\right]
\end{align}
Here the subscripts J on the $\overline{W}_m^{\rm J}$ indicate the connected correlations with respect to the Jacobi $\beta$ ensemble, $\kappa := \beta/2$ and
$J_n := \{ x_2,\dots, x_n \}$ with $J_1 = \emptyset$. Also, the tuple $(x_1,\ldots,\hat{x}_k,\ldots,x_n)$ consists of $x_1,\dots,x_n$ in order with $x_k$ excluded.

The mapping (\ref{1.6b+}) allows the loop equations (\ref{4.5}) for the Jacobi $\beta$ ensemble to be transformed to the loop equations for the circular Jacobi $\beta$ ensemble with weight (\ref{1.6c}) by a simple substitution of $\lambda_1, \lambda_2$. With this done, the large $N$ expansion (\ref{E2d}) can be substituted and like powers of $N$ equated to obtain a triangular system of equations for $\{W_n^{(l)}\}$.

\begin{corollary}
    Let $\kappa = \beta/2$,
    $\tilde{\lambda}_1 = \kappa - 1 - \kappa p - qi$, $\tilde{\lambda}_2 = 2 \kappa p$. We have
    \begin{align}
    -{1 \over x}& \overline{W}_{1}^{0}(x)+(\overline{W}_{1}^{0}(x))^2 = 0, \\
    - {\kappa \over x_1} & \overline{W}_{2}^{0}(x_1,x_2)-{1 \over x_1  (1 - x_1)}\overline{W}_{1}^{0}(x_2)- {1 \over x_1 (1 - x_1)} x_2 {\partial \over \partial x_2} \overline{W}_{1}^{0}(x_2) \\
     \qquad & + {\partial \over \partial x_2} \bigg ( 
     {\overline{W}_{1}^{0}(x_1) - \overline{W}_{1}^{0}(x_2) \over x_1 - x_2} + {1 \over x_1} \overline{W}_{1}^{0}(x_2) \bigg ) + 2\kappa
    \overline{W}_{1}^{0}(x_1) \overline{W}_{2}^{0}(x_1,x_2) = 0, \nonumber 
    \end{align}
    \begin{align}
    (\kappa - 1) & {\partial \overline{W}_{1}^{0}(x) \over \partial x} + \Big ( {\tilde{\lambda}_1 \over x} - {\tilde{\lambda}_2 \over 1 - x} \Big )
    \overline{W}_{1}^{0}(x) - {\kappa \over x} \overline{W}_{1}^{1}(x) \\
     \qquad & +{\kappa p - i q \over x(1 - x)} + \kappa \overline{W}_{1}^{0}(x,x) +2 \kappa \overline{W}_{1}^{0}(x) \overline{W}_{1}^{1}(x)=0, \nonumber \\
	-\frac{\kappa}{x_1}\overline{W}_{1}^{l+1}(x_1)&+\left[(\kappa-1)\frac{\partial}{\partial x_1}+\frac{\tilde{\lambda}_{1}}{x_1}-\frac{\tilde{\lambda}_2}{1-x_1}\right]\overline{W}_1^{l}(x_1)  \\\nonumber
	&+\kappa\overline{W}_2^{l-1}(x_1,x_1)+\kappa\sum_{m=0}^{l+2}\overline{W}_{1}^{m}(x_1)\overline{W}_{1}^{l+2-m}(x_1)=0, \quad  l\geq 1, 
\end{align}
and
\begin{align}\label{4.10}
	 -\frac{\kappa}{x_1}&\overline{W}_{2}^{l+1}(x_1,x_2)+\left[(\kappa-1)\frac{\partial}{\partial x_1}+\frac{\tilde{\lambda}_{1}}{x_1}-\frac{\tilde{\lambda}_2}{1-x_1}\right]\overline{W}_2^{l}(x_1,x_2)\\\nonumber
	&-\frac{1}{x_1(1-x_1)}\left(x_2\frac{\partial}{\partial x_2}+1\right)\overline{W}_1^{l+1}(x_2)+\frac{\partial}{\partial x_2}\left\{\frac{\overline{W}_1^{l+1}(x_1)-\overline{W}_1^{l+1}(x_2)}{x_1-x_2}+\frac{1}{x_1}\overline{W}_1^{l+1}(x_2)\right\}\\\nonumber
	&+\kappa\overline{W}_3^{l-1}(x_1,x_1,x_2)+2\kappa\sum_{m=0}^{l+1}\overline{W}_{1}^{m}(x_1)\overline{W}_2^{l+1-m}(x_1,x_2)=0, \quad  l\geq 0.
\end{align}
In addition, an equation analogous to (\ref{4.10}) holds which expresses $\overline{W}_{n}^{l+1}$ for $n \ge 3$ in terms of lower order terms in the triangular system.
\end{corollary}
  
A hand calculation readily suffices to solve the above triangular system up to and including $\overline{W}_{1}^{2}(x)$. Thus we find
\begin{equation}\label{iS1}
\overline{W}_{1}^{0}(x)=\frac{1}{x},\quad \overline{W}_{1}^{1}(x)=\frac{\kappa p+iq}{\kappa x(1-x)},\quad
	\overline{W}_{1}^{2}(x)=\frac{(\kappa p-iq-\kappa+1)(\kappa p+iq)}{\kappa^2(x-1)^2},
\end{equation}  
with the intermediate results
\begin{equation}\label{iS2}
\overline{W}_{2}^{0}(x,y)=\overline{W}_{2}^{1}(x,y)=0, \qquad
\overline{W}_{3}^{0}(x,y,z)=0.
\end{equation}
The evaluations of $\overline{W}_{1}^{0}(x), \overline{W}_{1}^{1}(x)$ and $\overline{W}_{2}^{0}(x,y)$ can in fact be obtained by specialising parameters in results from \cite[Prop.~4.8]{FRW17}.
The first of the equations in (\ref{iS1}) is equivalent to the normalisation
\begin{equation}\label{iS3}
\int_0^{2\pi} \rho_{(1),N}^{(\widetilde{\rm cJ})}(\theta;\beta,p,q) \, d \theta = N = { N \over 2 \pi} \delta_{k=0} \int_0^{2 \pi} e^{ i k \theta} \, d \theta, \quad k \in \mathbb Z \backslash\{ 0 \}.
\end{equation}
As commented in the sentence above (\ref{1.2f}), this last expression, with the restriction to $k = 0$ removed, is to be subtracted from the Fourier coefficient $\int_0^{2\pi} \rho_{(1),N}^{(\widetilde{\rm cJ})}(\theta;\beta,p,q) e^{ i k \theta} \, d \theta$ to arrive at (\ref{1.2g+}). With this done,
according to the definition of $\{ \alpha_l \}$ given above (\ref{E2a}), we read off from the
second and third  equations in (\ref{iS1}) that
\begin{equation}\label{iS4}
\alpha_0 = - p - \frac{i q}{\kappa},
\qquad  \alpha_1=
\left(\frac{1}{\kappa }-1\right) p+p^2+\frac{q(-q+(\kappa-1)i)}{\kappa ^2},
\end{equation}
where $\kappa := \beta/2$.

The use of computer algebra allows the evaluations (\ref{iS2}) to be extended, and thus similarly for (\ref{iS4}). Specifically, in relation to the latter, this shows
\begin{align}
	 \alpha_2=&\frac{\alpha(1+\bar{\alpha}-\kappa)(-2+\alpha-\bar{\alpha}+2 \kappa)}{2\kappa^3},\\
	 \alpha_3=&\frac{\alpha(1+\bar{\alpha}-\kappa)\left(6+\alpha^2+\bar{\alpha}^2-3 \alpha(2+\bar{\alpha}-2 \kappa)-5 \bar{\alpha}(\kappa-1)-11 \kappa+6 \kappa^2\right)}{6\kappa^4},
  \end{align}
  \begin{align}
	 \alpha_4 =&\frac{1}{24\kappa^5} \alpha(1+\bar{\alpha}-\kappa)\left(\alpha^3-\bar{\alpha}^3-6 \alpha^2(2+\bar{\alpha}-2 \kappa)+9 \bar{\alpha}^2(\kappa-1)+\bar{\alpha}\left(-26+47 \kappa-26 \kappa^2\right)+\right.  \\
		& \left.\quad \alpha\left(34+6 \bar{\alpha}^2-29 \bar{\alpha}(\kappa-1)-63 \kappa+34 \kappa^2\right)+12\left(-2+5 \kappa-5 \kappa^2+2 \kappa^3\right)\right),\nonumber\\\label{al5}
	 \alpha_5=&\frac{1}{120\kappa^6} \alpha(1+\bar{\alpha}-\kappa)\left(\alpha^4+\bar{\alpha}^4-10 \alpha^3(2+\bar{\alpha}-2 \kappa)-
  14 \bar{\alpha}^3(\kappa-1)\right. \\
		& +5 \alpha^2\left(22+4 \bar{\alpha}^2-19 \bar{\alpha}(\kappa-1)-41 \kappa+22 \kappa^2\right)+ \bar{\alpha}^2\left(71-127 \kappa+71 \kappa^2\right)\nonumber \\
		& +\bar{\alpha}\left(154-377 \kappa+377 \kappa^2-154 \kappa^3\right)+ 4\left(30-91 \kappa+124 \kappa^2-91 \kappa^3+30 \kappa^4\right) \nonumber \\
		& -5 \alpha\left(42+2 \bar{\alpha}^3- 
		\left.17 \bar{\alpha}^2(\kappa-1)-107 \kappa+107 \kappa^2-42 \kappa^3+\bar{\alpha}\left(47-86 \kappa+47 \kappa^2\right)\right)\right),
  \nonumber
\end{align}
where $\alpha:= \kappa p + i q$. We are now in a position to specify the coefficients in (\ref{1.h}) up to and including index label $j=5$.

\begin{proposition}\label{P4.2}
 Let $\alpha_0,\dots, \alpha_5$ be given by (\ref{iS4})--(\ref{al5}). We have
 \begin{equation}
     h_j(\beta,p,q) = {i \over (2 \pi)^j }{\rm Im} \, \alpha_j, \: \: j \: {\rm even} \qquad 
 h_j(\beta,p,q) = {1 \over (2 \pi)^j }{\rm Re} \, \alpha_j, \: \: j \: {\rm odd}    
 \end{equation}
 and
 \begin{equation}\label{4.23}
     \tilde{h}_j(\beta,p,q) = {1 \over (2 \pi)^j }{\rm Re} \, \alpha_j, \: \: j \: {\rm even} \qquad 
 \tilde{h}_j(\beta,p,q) = {i \over (2 \pi)^j }{\rm Im} \, \alpha_j, \: \: j \: {\rm odd} .   
 \end{equation}
 Explicity, for indices up to and including $j=3$,
 \begin{align}
&h_0(\beta,p,q) = -{iq \over \kappa}, \quad
h_1(\beta,p,q) = {1 \over 2 \pi \kappa^2} (q^2 + \kappa p + \kappa^2(-1+p)p) , \label{4.24} \\ & h_2(\beta,p,q) =
{i q \over 4 \pi^2 \kappa^3} (-1+q^2+ \kappa(2+p) + \kappa^2(-1-p+p^2)), \\
 &h_3(\beta,p,q) =
{1 \over 48 \pi^2 \kappa^4}\Big (
   17 q^2 - 5 q^4 
 + \kappa (-33 q^2 - 6 p (-1 + q^2)) + 
  \kappa^2 (17 q^2 + p^2 (5 - 6 q^2)\\ & \qquad \qquad    +  p (-17 + 6 q^2))+
\kappa^3 p (17 - 9 p - 2 p^2) - \kappa^4 p (6 - 5 p - 2 p^2 + p^3) 
  \Big )\nonumber
\end{align}
and
\begin{align}\label{4.27}
&\tilde{h}_0(\beta,p,q) = - p, \: \: \tilde{h}_1(\beta,p,q) = {iq (1 - \kappa) \over 2 \pi \kappa^2}, \: \:
\tilde{h}_2(\beta,p,q) = -{(1 - \kappa)  \over 4 \pi^2 \kappa^3} (2 q^2 + \kappa p + \kappa^2(-1+p)p) ,  \\ \label{4.28} & 
\tilde{h}_3(\beta,p,q) = -{iq (1 - \kappa)  \over 48 \pi^3 \kappa^4}
(-6 +  
   16 q^2 +  \kappa  (11 + 12 p) + 6 \kappa^2 (-1 - 2 p + 2 p^2) ).
\end{align}
\end{proposition}

\begin{proof}
    This follows by comparing (\ref{1.h}) and (\ref{E2a}), then making use of the property noted in the sentence below (\ref{1.2gz+}).
\end{proof}

\begin{remark}
1.~Based on the loop equations (\ref{4.5}) it has been shown in \cite{Fo22} that the Jacobi $\beta$ ensemble connected correlators $\overline{W}_n^{\rm J}$ satisfy the functional equation
$$
\overline{W}_n^{\rm J}(x_1,\dots,x_n;N,\kappa,\lambda_1, \lambda_2) =
(-\kappa)^{-n}\overline{W}_n^{\rm J}(x_1,\dots,x_n;N,1/\kappa,-\kappa \lambda_1,-\kappa \lambda_2),
$$
where as in Proposition \ref{P4.2}, $\kappa := \beta/2$.
The mapping (\ref{1.6b}) then induces the functional equation for the circular Jacobi $\beta$ ensemble connected correlators $\overline{W}_n^{\widetilde{\rm cJ}}$,
$$
\overline{W}_n^{\widetilde{\rm cJ}}
(x_1,\dots,x_n;N,\kappa,p,q) = (-\kappa)^{-n} \overline{W}_n^{\widetilde{\rm cJ}}
(x_1,\dots,x_n;-\kappa N,1/\kappa,-\kappa p,q/\kappa).
$$
Now setting $n=1$, expanding in $1/N$ as in (\ref{E2}), then proceeding from this to (\ref{E2a}) gives the functional equation
$$
c_\infty^{\widetilde{\rm cJ}}
(\tau;\beta,p,q) = - {2 \over \beta}
c_\infty^{\widetilde{\rm cJ}}
(-2\tau/\beta;4/\beta,-\beta p/2,-2q/\beta)
$$
or equivalently
\begin{equation}\label{4.30}
     h_j
(\beta,p,q) = \Big ( - {2 \over \beta} \Big )^{j+1}
h_j
\Big ({4 \over \beta},-{\beta p \over 2},-{ 2q \over \beta}\Big ), \quad \tilde{h}_j
(\beta,p,q) = \Big ( - {2 \over \beta} \Big )^{j+1}
\tilde{h}_j
\Big ({4 \over \beta},-{\beta p \over 2},-{2q \over \beta} \Big ).
\end{equation}
The functional equations (\ref{4.30}) are readily exhibited on the explicit functional forms (\ref{4.24})--(\ref{4.28}). \\
2.~In the case $q=0$ the density $\rho_{(1),\infty}^{(\rm cJ)}( x ;\beta,p,q)$ is an even function of $x$ and consequently $c_\infty^{(\widetilde{\rm cJ})}(\tau;\beta,p,q)$ is an even function of $\tau$. In the expansion (\ref{1.h}) we must then have that $h_j(\beta,p,q)|_{q=0}=0$ for $j$ even and $\tilde{h}_j(\beta,p,q)|_{q=0}=0$ for $j$ odd, which can also be viewed as corollaries of (\ref{1.2gz}). This is indeed a feature of the explicit functional forms (\ref{4.24})--(\ref{4.28}).
\end{remark}

\section{Discussion}\label{S5}
Setting $q=0$ we read off from (\ref{1.2g+}), (\ref{1.h}) and the results of Proposition \ref{P4.2} that
\begin{multline}\label{5.1}
\int_{-\infty}^\infty \Big ( \rho_{(1),\infty}^{(\rm cJ)}(x;\beta,p,q) |_{q=0} - 1 \Big ) e^{i \tau x} \, dx \\ = -p + {p \over 2 \pi}((1-\kappa) + \kappa p) {|\tau| \over \kappa} -
{p (1 - \kappa) \over 4 \pi^2}((1 - \kappa) + \kappa p) {\tau^2 \over \kappa^2} + {\rm O}(\tau^3).
\end{multline}
It is of interest to compare this expansion to the analogue for the two-dimensional one-component plasma (\ref{1.34}). With $Q$ identified as $p$, we see that the first term on the RHS, which corresponds to the perfect screening of the external charge, is the same for both. In two-dimensions the generalised Fourier transform of the logarithmic potential $V^{\rm OCP}(\mathbf r) = - \log | \mathbf r |$ is $\hat{V}^{\rm OCP}(\mathbf k) = 2 \pi/ |\mathbf k|^2$ and the dimensionless thermodynamic pressure is $\beta P^{\rm OCP} = 1 - \beta/4$. Hence the second term in the expansion (\ref{1.34}) can be written
\begin{equation}\label{qw1}
{Q \over \hat{V}^{\rm OCP}(\mathbf k) \beta} (
\beta P^{\rm OCP} + (\beta /4)Q).
\end{equation}
In relation to the circular $\beta$ ensemble (C$\beta$E), the generalised Fourier transform of the logarithmic potential $V^{{\rm C}\beta{\rm E}}(r) = - \log |  r |$ is $\hat{V}^{{\rm C}\beta{\rm E}}(\tau) =  \pi/ |\tau|$ and the dimensionless thermodynamic pressure is $\beta P^{{\rm C}\beta{\rm E}} = 1 - \beta/2$. Hence the second term in the expansion (\ref{5.1}) can be written
\begin{equation}\label{qw2}
{p \over \hat{V}^{{\rm C}\beta{\rm E}}(\tau) \beta} (
\beta P^{{\rm C}\beta{\rm E}} + (\beta /2)p).
\end{equation}
The structural similarity between (\ref{qw1}) and (\ref{qw2}) is evident. However when it comes to comparing the third terms on the RHSs of (\ref{5.1}) and (\ref{1.34}) structural differences show. Specifically the former is $Q$ times a quadratic polynomial in $Q$, while the latter is $p$ times a linear polynomial in $p$. This is perhaps not unsurprising as the argument used in \cite{Sa19} to predict the $Q^2$ term relies in part on the screening cloud having a fast decay with well defined moments, which is not true in one-dimension. The dependence on $\beta$ is quadratic in $1/\beta$ in both cases, although in (\ref{5.1}) there is a stand alone factor of $(1/\kappa - 1)$ with no analogue in (\ref{1.34}).

A linear response argument is used in \cite{Sa19} to deduce that
\begin{multline}
    \lim_{Q \to 0} {1 \over Q} \int_{\mathbb R^2} ( \rho_{(1),\infty}^{\rm OCP}(\mathbf r;Q) - 1) e^{i \mathbf k \cdot \mathbf r} \, d \mathbf r = 
    -\beta \tilde{V}^{\rm OCP}(\mathbf k) S_\infty^{\rm OCP}(\mathbf k;\beta) \\ =
   - \beta \tilde{V}^{\rm OCP}(\mathbf k) \Big (1 + 
    \int_{\mathbb R^2} ( \rho_{(1),\infty}^{\rm OCP}(\mathbf r;Q)|_{Q=1} - 1) e^{i \mathbf k \cdot \mathbf r} \, d \mathbf r \Big ),
\end{multline}
where the second equality follows from the two-dimensional analogue of (\ref{1.3c}). As noted in \cite{Sa19}, this sum rule is readily checked to be a feature of the expansion (\ref{1.ps}). The same linear response argument can be applied to the one-dimensional log-gas to obtain
\begin{multline}\label{5.5}
    \lim_{p \to 0} {1 \over p} \int_{-\infty}^\infty ( \rho_{(1),\infty}^{\rm OCP}(x;p,q)|_{q=0} - 1) e^{i \tau x} \, d x = 
    -\beta \tilde{V}^{{\rm C}\beta{\rm E}}(\tau) S_\infty^{{\rm C}\beta{\rm E} }(\tau;\beta) \\ =
   - \beta \tilde{V}^{{\rm C}\beta{\rm E}}(\tau) \Big (1 + 
   \int_{-\infty}^\infty ( \rho_{(1),\infty}^{\rm OCP}(x;p,q)|_{p=1\atop q=0} - 1) e^{i \tau x} \, d x
     \Big ),
\end{multline}
or equivalently, using (\ref{1.2g}) and (\ref{1.3c}), the relation
\begin{equation}\label{5.6}
    \lim_{p \to 0} {1 \over p} c_\infty^{(\widetilde{\rm cJ)}}(\tau;\beta,p,0) = - {\pi \beta \over |\tau|} \Big ( 1 + 
 c_\infty^{(\widetilde{\rm cJ)}}(\tau;\beta,1,0) \Big ).
\end{equation}
Using the results of Proposition \ref{P4.2} we can check the validity of (\ref{5.6}) in a small $\tau$ expansion up to and including order $\tau^4$. 

Linear response also applies in the limit $q \to 0$. To facilitate this,
define $V(x) = \pi {\rm sgn}(x)$. Noting as an improper integral that
$$
i \int_{-\infty}^\infty {e^{-ix \tau} \over \tau} \, d \tau = \pi {\rm sgn}(x)
$$
allows us to take the inverse Fourier transform to conclude $\hat{V}(\tau) = 2 \pi i/\tau$. 
With this knowledge we can deduce the analogues of (\ref{5.5}) and (\ref{5.6}). Specifically, in relation to the latter we obtain
\begin{equation}\label{5.7}
    \lim_{q \to 0} {1 \over q} c_\infty^{(\widetilde{\rm cJ)}}(\tau;\beta,0,q) = - {2 i  \over \tau} \Big ( 1 + 
 c_\infty^{(\widetilde{\rm cJ)}}(\tau;\beta,1,0) \Big ).
\end{equation}
As with (\ref{5.6}),
using the results of Proposition \ref{P4.2} we can check the validity of (\ref{5.7}) in a small $\tau$ expansion up to and including order $\tau^4$. 

In Proposition \ref{P4.2} let us  replace $q$ by $\kappa q$. Then $h_j$ and $\tilde{h}_j$ are each polynomials of degree $j$ in $1/\kappa$. As remarked in the Introduction below (\ref{1.17b}), in the case $p=1$, $q=0$ these polynomials have the property of being palindromic or anti-palindromic in $u:=1/\kappa$. Moreover, it has been observed from the explicit form of these polynomials that all the zeros lie on the unit circle $|u|=1$ in the complex $u$-plane and moreover exhibit an interlacing property \cite{FJM00,Fo21a}, which is a feature too of the polynomials appearing in the $1/N$ expansion of the moments of the spectral density for the Gaussian $\beta$ ensemble \cite{WF14}. (More on this has been communicated to the senior author by Michael A. La Croix, who among other things highlights the references \cite{DF17,GJ96,LC09}; one should add too the recent work \cite{No22}, although any sought of explanation  is still lacking.)
Our results of Proposition \ref{P4.2} reveal these same properties for the polynomials in $u:=1/\kappa$ which result as coefficients by expanding $h_j$ or $\tilde{h}_j$ in a power series in $p$ for $q=0$, or in $q$ for $p=0$. For example, from (\ref{al5}) and (\ref{al5}) we compute that in relation to $h_5(\beta,p,q)|_{q=0}$ the coefficients of $p,\dots,p^5$ in order are, up to proportionality, the (anti-)palindromic polynomials
\begin{align*}
    &(u-1)(30-91 u + 124 u^2 - 91 u^3 + 30 u^4), \quad
32-75u+90u^2-75u^3+32u^4, \\
&(u-1)(11-20u + 11u^2), \quad 5-9u+5u^2, \quad u-1,
\end{align*}
and which indeed can be checked to have successively interlacing zeros all on the unit circle in the complex $u$-plane.

The results of Proposition \ref{P4.2} take on a simpler form in the case $q=0$ with $\beta\to \infty$ (low temperature limit),
\begin{multline}
    \lim_{\beta \to \infty} c_\infty^{(\widetilde{\rm cJ})}(\tau;\beta,p,0) = - p +\tilde{p} {|\tau| \over 2 \pi} +
  \tilde{p} \Big ({\tau \over 2 \pi} \Big )^2 +
 \tilde{p} \Big ( 1 - {1 \over 6}  \tilde{p} \Big ) 
 \Big ( {|\tau| \over 2 \pi} \Big )^3 \\
 + \tilde{p} \Big ( 1 - {1 \over 3}  \tilde{p} \Big ) 
 \Big ( {\tau \over 2 \pi} \Big )^4 +
\tilde{p} \Big ( 1 - {7 \over 15}  \tilde{p} + {1 \over 60} \tilde{p}^2 \Big ) 
 \Big ( {\tau \over 2 \pi} \Big )^5 + \cdots, 
\end{multline}
where $\tilde{p} := p(p-1)$. The factor of $\tilde{p}$ in all terms but the first makes sense as the ground state is an equally spaced lattice for both $p=0$ and $p=1$, which is when $\tilde{p}$ vanishes. Also, the leading term in $\tilde{p}$ of each power of $\tau$ is suggestive of the summation in $\tau$
\begin{equation}\label{eqc}
\lim_{\beta \to \infty} c_\infty^{(\widetilde{\rm cJ})}(\tau;\beta,p,0) + p
\mathop{\sim}\limits_{\tilde{p} \to 0}
{\tilde{p} |\tau|/(2 \pi) \over 1 - |\tau|/(2 \pi)},
\end{equation}
which makes explicit the form of the singularity in $\tau$ as $|\tau| \to (2\pi)^-$. In fact for $p \to 0$ (and thus $\tilde{p} \to -p \to 0$) this can be anticipated from the first equality in (\ref{5.5}). The required input is the large $\beta$ form of $S_\infty^{{\rm C} \beta {\rm E}}(\tau;\beta)$, which we know from \cite[\S 5]{FJM00} to be given by the RHS of (\ref{eqc}) with $\tilde{p}$ replaced by $2/\beta$. Substituting in the first equality of (\ref{5.5}) gives a result equivalent to (\ref{eqc}) for $p \to 0$.

\begin{remark}
With $q=0$ a result of \cite{FR86} tells us that with bulk spectrum singularity scaling, in the limit $\beta \to \infty$ the eigenvalues crystallise at the zeros of the Bessel function $J_{p-1/2}(\pi x)$.
\end{remark}
 \subsection*{Acknowledgements}
 The contribution of Menglin Wang in helping with some of the calculations in Section \ref{S2} is acknowledged. 
	This research is part of the program of study supported
	by the Australian Research Council 
	 Discovery Project grant DP210102887. In particular the grant partially supported the visit of Bo-Jian Shen to the University of Melbourne to work on this project.
  The research of Bo-Jian Shen is also supported by the National Natural Science Foundation of China (Grant No.12175155) and the Shanghai Frontier Research Institute for Modern Analysis.


\begin{thebibliography}{10}
\bibitem{AAW23}
G.~Akemann, N.~Ayg\"un and T.R.~W\"urfel,
\emph{Generalised unitary group integrals of Ingham-Siegel and
Fisher-Hartwig type}, arXiv:2305.19852.
\bibitem{ABGS20}
T.~Assiotis, B.~Bedert, M.A.~Gunes and A.~Soor,
\textit{Moments of generalised Cauchy random matrices and continuous-Hahn polynomials},
Nonlinearity, \textbf{34} (2021), 4923.

\bibitem{AGS21}
T.~Assiotis, M.A.~Gunes and A.~Soor,
\emph{Convergence and an explicit formula for the joint moments of the circular
Jacobi $\beta$-ensemble characteristic polynomial},
Math. Phys. Anal. Geom. \textbf{25} (2022), 15.

\bibitem{BWW18}
N. Berestycki, C. Webb, and M.D. Wong, \emph{Random Hermitian matrices and Gaussian multiplicative chaos}, Probab. Theory
Related Fields \textbf{172} (2018),  103--189.

\bibitem{BO01}
A.~Borodin, and G. Olshanski. \emph{Infinite random matrices and ergodic measures},
 Comm. Math. Phys. \textbf{223} (2001), 87--123.
 
 \bibitem{BG11}
G.~Borot, A.~ Guionnet,  \emph{Asymptotic expansion of $\beta$ matrix models in
  the one-cut regime}, Commun. Math. Phys. 317 (2013), 447--483.

  \bibitem{BF22}
P. Bourgade and H. Falconet, \emph{Liouville quantum gravity from random matrix dynamics},
arXiv 2206.03029.

\bibitem{BC22}
S.-S. Byun and C. Charlier, \emph{On the characteristic polynomial of the eigenvalue moduli of random normal matrices}, 
arXiv:2205.04298.
 
  \bibitem{BF22a}
   S.-S.~Byun and P.J.~Forrester,   \emph{Progress on the study of the Ginibre ensembles I: G{\SMALL in}UE},
   arXiv:2211.16223.


       \bibitem{CMS17}
A. del Campo, J. Molina-Vilaplana and J. Sonner, \emph{Scrambling the spectral form factor:
unitarity constraints and exact results}, Phys. Rev. D \textbf{95} (2017), 126008. 

\bibitem{CLW15}
 T. Can, M. Laskin and P. Wiegmann, \emph{Geometry of quantum Hall states: gravitational anomaly and kinetic
coefficients}, Annals Phys. \textbf{362} (2015) 752--794.

\bibitem{CG21}
C. Charlier and R. Gharakhloo, \emph{Asymptotics of Hankel determinants with a Laguerre-type or Jacobi-type potential and Fisher-Hartwig singularities}, Adv. Math. \textbf{383} (2021), 107672.

\bibitem{CES21}
C.~Cipolloni, L.~Erd\"os and D.~Schr\"oder,
\emph{On the spectral form factor for random matrices}, 
Commun. Math. Phys. (2023). https://doi.org/10.1007/s00220-023-04692-y

\bibitem{CGMY22}
T. Claeys, G. Glesner, A. Minakov and M. Yang, \emph{Asymptotics for averages over classical orthogonal ensembles},
Int. Math. Res. Not. \textbf{2022} (2022), 7922--7966.

\bibitem{C+17} J.S. Cotler, G. Gur-Ari, M. Hanada, J. Polchinski, P. Saad, S.H. Shenker, D. Stanford,
A. Streicher and M. Tezuka, \emph{Black Holes and Random Matrices}, JHEP \textbf{1705} (2017), 118;
Erratum: [JHEP \textbf{1809} (2018), 002]

\bibitem{CHLY17}
J.S. Cotler, N. Hunter-Jones, J. Liu and B. Yoshida, \emph{Chaos, complexity, and random
matrices}  JHEP \textbf{1711} (2017), 048 

\bibitem{CCO20}
P. Cohen, F. Cunden and N. O'Connell,
Moments of discrete orthogonal polynomial ensembles. Electron. J. Probab. 25 (2020), 1-19.

\bibitem{CMOS19}
F. Cunden, F. Mezzadri, N. O'Connell and N. Simm,
 Moments of random matrices and hypergeometric orthogonal polynomials. Comm. Math. Phys. \textbf{369}  (2019), 1091-1145.

\bibitem{DS22}
 A. Dea\~{n}o and N. Simm, \emph{Characteristic polynomials of complex random matrices and Painlevé transcendents}, Int. Math. Res.
Not., \textbf{2022} (2022), 210--264.

 \bibitem{DIK13}  P.~Deift, A.~Its and I.~Krasovsky, \emph{Toeplitz matrices and Toeplitz determinants
under the impetus of the Ising model: some historic and some recent results}, Comm. Pure Appl. Math.
\textbf{66} (2013), 1360--1438.

\bibitem{DKV11}
P. Deift, I. Krasovsky and J. Vasilevska.
 \emph{Asymptotics for a determinant with a confluent hypergeometric kernel},
 Int. Math Res. Not.,  {\bf 2011}  (2011), 2117--2160.

\bibitem{DF17}
M. Do{\l}ga and Valentin F\'eray, \emph{Cumulants of Jack symmetric functions and the b-conjecture},
Trans. Amer. Math. Society \textbf{369}  (2017), 9015--9039.

\bibitem{FH68}
M.E. Fisher and R.E. Hartwig, \emph{Toeplitz determinants -- some applications,
  theorems and conjectures}, Adv. Chem. Phys. \textbf{15} (1968), 333--353.


 \bibitem{Fo10}
P.J.~Forrester, \emph{Log-gases and random matrices}, Princeton University Press,
  Princeton, NJ, 2010.
  


\bibitem{Fo21a}
  P.J. Forrester, \emph{Differential identities for the structure function of some random matrix
  ensembles}, J.~Stat.~Phys. \textbf{183} (2021), 33.

   \bibitem{Fo21b}
 P. J. Forrester, \emph{Quantifying dip-ramp-plateau for the Laguerre unitary ensemble structure function}, 
 Commun. Math. Phys. \textbf{387} (2021), 215--235.
  
   \bibitem{Fo21}
P.J.~Forrester, \emph{Joint moments of a characteristic polynomial and its derivative for the
circular $\beta$ ensemble}, Probab. Math. Phys. \textbf{3} (2022), 145--170.   
  
\bibitem{Fo22}
 P.J. Forrester, \emph{High-low temperature dualities for the classical
$\beta$-ensembles}, Random Matrices Th. Appl. \textbf{11},  (2022)   2250035. 
  

 \bibitem{FJM00}
P.J. Forrester, B.~Jancovici, and D.S. McAnally, \emph{Analytic properties of
  the structure function for the one-dimensional one-component log-gas}, J.
  Stat. Phys. \textbf{102} (2000), 737--780.

   \bibitem{FKLZ23}
P.J. Forrester, M.~Kieburg, S.-H. Li and J.~Zhang, \emph{Dip-ramp-plateau for  Dyson Brownian motion from the identity on $U(N)$ }, arXiv:2206.14950.

 \bibitem{FK22} P.J. Forrester and S.~Kumar, \emph{Differential recurrences for the distribution of the trace of the
$\beta$-Jacobi ensemble}, Physica D \textbf{434} (2022), 133220.

\bibitem{FLSY23}
P.J.~Forrester, S.-H.~Li, B.-J.~Shen and G.-F. Ye, \emph{$q$-Pearson pair and moments in $q$-deformed ensembles}, Ramanujan J. \textbf{60} (2023), 195--235.

 \bibitem{FLT21}
P.J. Forrester, S.-H. Li and A.K.~Trinh, \emph{Asymptotic correlations with corrections for the
circular Jacobi $\beta$-ensemble}, J.~Approx. Th.  \textbf{271} (2021), 105633.

\bibitem{FN01}
P.J. Forrester and T. Nagao.
\newblock \emph{Correlations for the Cauchy and generalized circular ensemble with orthogonal and symplectic symmetry}.
\newblock J. Phys. A \textbf{34} (2001), 7919-7932.
  
  \bibitem{FR21}
P.J. Forrester and A.A. Rahman,  \emph{Relations between moments for the Jacobi and Cauchy random matrix ensembles}, J. Math. Phys. \textbf{62} (2021), 073302.

 \bibitem{FRW17}
P.J. Forrester, A.A. Rahman, and N.S. Witte, \emph{Large $N$ expansions for the Laguerre and Jacobi $\beta$ ensembles from the loop equations}, J. Math. Phys. \textbf{58}
  (2017), 113303.

  \bibitem{FR86}
P.J. Forrester and J.B. Rogers, \emph{Electrostatics and the zeros of the
  classical polynomials}, SIAM J. Math. Anal. \textbf{17} (1986), 461--468.

\bibitem{FS16}
  Y.V.~Fyodorov and N.J.~Simm, \emph{On the distribution of the maximum value of the characteristic polynomial
of GUE random matrices}, Nonlinearity \textbf{29} (2016), 2837--2855.

  \bibitem{GGR21}
M. Gisonni, T. Grava  and G. Ruzza, \emph{Jacobi ensemble, Hurwitz numbers and Wilson polynomials}, Lett. Math. Phys. 
111 (2021), 67.

\bibitem{GJ96}
I.P. Goulden and D.M. Jackson, \emph{Connection coefficients, matchings, maps and combinatorial
conjectures for Jack symmetric functions}, Trans. Amer. Math. Soc. \textbf{348} (1996), 873--892.

   \bibitem{HZ86}
J.~Harer and D.~Zagier,  The {E}uler characteristic of the moduli space of
	curves. Inven. Math.  \textbf{85} (1986), 457--485.

  
\bibitem{Hu63}
L.K. Hua,
\emph{Analysis of Functions of Several Complex Variables in the
Classical Domains},
American Mathematical Society, Providence, RI,
1963.  

\bibitem{JS08}
B.~Jancovici and L.~\v{S}amaj, 
\emph{Guest charge and potential fluctuations in two-dimensional classical Coulomb systems},
 J. Stat. Phys. \textbf{131} (2008), 613--629.
 
 
\bibitem{KMST00} 
 P.~Kalinay, P.~Marko\v{s},  L. \v{S}amaj and I.~Trav\v{e}nec,  \emph{The sixth-moment sum rule for the pair correlations of the two-dimensional one-component plasma: Exact result},
  J. Stat. Phys.   \textbf{98} (2000), 639--666. 

\bibitem{LC09}
  M.A. La Croix, \emph{The combinatorics of the Jack parameter and the genus series for topological
maps}, Ph.D. thesis, University of Waterloo, 2009.

  \bibitem{Le04}
M.~Ledoux,
\emph{Differential operators and spectral distributions of invariant ensembles
	from the classical orthogonal polynomials. The continuous case},
Electron. J. Probab. \textbf{9} (2004), 177--208.

\bibitem{Le09}
M.~Ledoux,
\emph{A recursion formula for the moments of the {G}aussian orthogonal
	ensemble}, Ann. Inst. Henri Poincar\'e Probab. Stat. \textbf{45} (2009),
754--769.
  
  \bibitem{Li58}
J. Lighthill,
\emph{Introduction to Fourier analysis and generalized functions},
Cambridge University Press, {1958}.


\bibitem{Li17}
D.-Z.~Liu,
\newblock \emph{Limits for circular Jacobi beta-ensembles},
\newblock J. Approx. Theory, \textbf{215} (2017), 40--67.


\bibitem{MS13}
F. Mezzadri and N. J. Simm, $\tau$-function theory of
quantum chaotic transport with $\beta=1,2,4$, Commun.
Math. Phys. \textbf{324} (2013), 465--513.

 \bibitem{MH21}
 A.~Mukherjee and S.~Hikami, \emph{Spectral form factor for time-dependent matrix model}, JHEP \textbf{2021} (2021) 071. 

 \bibitem{NS93}
T.~Nagao and K.~Slevin, \emph{Laguerre ensembles of random matrices:
  nonuniversal correlation functions}, J. Math. Phys. \textbf{34} (1993),
  2317--2330.

   \bibitem{DLMF}  NIST Digital Library of Mathematical Functions.

\bibitem{No22}
    M. Novaes, \emph{Time delay statistics for finite number of
channels in all symmetry classes}, Europhys. Lett. \textbf{139} (2022),
21001.

    \bibitem{Ok19}
  K. Okuyama, \emph{Spectral form factor and semi-circle law in the time direction},
  JHEP \textbf{2019} (2019), 161.   

\bibitem{RF21}
   A.A.~Rahman and P.J.~Forrester,  \emph{Linear differential equations for the resolvents of the classical matrix ensembles},
Random Matrices Theory Appl. \textbf{10} (2021), 2250003. 

\bibitem{Sa07}
L.~\v{S}amaj, \emph{A generalization of the Stillinger-Lovett sum rules for the two-dimensional jellium},
 J. Stat. Phys.  \textbf{128} (2007), 1415--1428.

\bibitem{Sa19}
L.~\v{S}amaj, \emph{Fourth moment of the charge density induced around a guest charge
in two-dimensional jellium}, J.~Stat.~Phys.  \textbf{175} (2019), 1066--1079.
  

 \bibitem{SR10}
 A.~Sri Ranga, \emph{Szeg\"o polynomials from hypergeometric functions}, Proc. Amer. Math. Soc.
 \textbf{138} (2010), 4243--4247.

   \bibitem{TGS18}
E.J.~Torres-Herrera, A.M.~Garc\'ia-Garc\'ia,  and L.F.~Santos, 
\emph{Generic dynamical features of quenched interacting quantum systems: 
Survival probability, density  imbalance,  and  out-of-time-ordered  correlator}, Phys. Rev. B
\textbf{97} (2018), 060303.

\bibitem{VG21} W.L. Vleeshouwers and V.~Gritsev, \emph{Topological field theory approach to intermediate
statistics}, SciPost Phys. \textbf{10} (2021), 146.


\bibitem{VG22} W.L. Vleeshouwers and V.~Gritsev,
\emph{The spectral form factor in the 't Hooft limit --- intermediacy versus universality}, 
SciPost Phys. Core \textbf{5} (2022), 051. 

\bibitem{WW19}
C.~Webb and M.D.~Wong, On the moments of the characteristic polynomial of a Ginibre random matrix, Proc.
Lond. Math. Soc. \textbf{118} (2019), 1017--1056.

    \bibitem{WF14}
N.S. Witte and P.J. Forrester, \emph{Moments of the {G}aussian $\beta$ ensembles
  and the large {$N$} expansion of the densities}, J. Math. Phys. \textbf{55}
  (2014), 083302.
 
 \bibitem{WF15}
N.S. Witte and P.J. Forrester, \emph{Loop equation analysis of the circular ensembles}, JHEP
  \textbf{2015} (2015), 173.

\bibitem{XZ20}
  S. Xu and Y. Zhao, \emph{Gap probability of the circular unitary ensemble with a Fisher-Hartwig
singularity and the coupled Painlev\'e V system}, Commun. Math. Phys. \textbf{377} (2020), 1545--1596.


\end{thebibliography}
\nopagebreak

\providecommand{\bysame}{\leavevmode\hbox to3em{\hrulefill}\thinspace}
\providecommand{\MR}{\relax\ifhmode\unskip\space\fi MR }
\providecommand{\MRhref}[2]{%
  \href{http://www.ams.org/mathscinet-getitem?mr=#1}{#2}
}
\providecommand{\href}[2]{#2}

\end{document}